\documentclass[aps, pra, a4paper,twocolumn,showpacs,10pt]{revtex4-1}

\usepackage{bbm,amsmath,amssymb,amsthm,bm, bbold,textcomp,graphicx,epsfig,color,verbatim,enumerate}

\newcommand{\ket}[1]{\left|#1\right\rangle}
\newcommand{\bra}[1]{\left\langle #1\right|}
\newcommand{\bracket}[2]{\left\langle #1|#2\right\rangle}
\newcommand\defn[1]{\textsl{#1}}

\newcommand\ketbra[2]{|#1\rangle\langle#2|}
\newcommand\cH{{\mathcal H}}

\newcommand\cB{{\mathcal B}}
\newcommand\cE{{\mathcal E}}

\newcommand{\one}{\mathbb{1}}

\def\sx{\sigma_x}
\def\sy{\sigma_y}
\def\sz{\sigma_z}

\def\j{{\bm j}}

\makeatletter
\newtheorem*{rep@theorem}{\rep@title}
\newcommand{\newreptheorem}[2]{%
\newenvironment{rep#1}[1]{%
 \def\rep@title{#2 \ref{##1}}%
 \begin{rep@theorem}}%
 {\end{rep@theorem}}}
\makeatother

\newreptheorem{theorem}{Theorem}
\newreptheorem{lemma}{Lemma}
\newtheorem{lemma}{Lemma}

\newtheorem{theorem}{Theorem}

\newtheorem{observation}{Observation}

\begin{document}

\title{Algebraic metrology: Pretty good states and bounds}
\author{Michael Skotiniotis$^1$, Florian Fr\"{o}wis$^{1,2}$, Wolfgang D\"{u}r$^1$, Barbara Kraus$^1$}

\affiliation{$^1$ Institut f\"ur Theoretische Physik, Universit\"at
  Innsbruck, Technikerstr. 25, A-6020 Innsbruck,  Austria.\\
  $^2$ Group of Applied Physics, University of Geneva, CH-1211 Geneva 4, Switzerland}
\date{\today}

\begin{abstract}
We investigate quantum metrology using a Lie algebraic approach for a class of Hamiltonians, including local and nearest-neighbor interaction Hamiltonians. Using this Lie algebraic formulation, we identify and construct highly symmetric states that admit Heisenberg scaling in precision in the absence of noise, and investigate their performance in the presence of noise. To this aim we perform a numerical scaling analysis, and derive upper bounds on the quantum Fisher information.
\end{abstract}
\pacs{03.67.-a, 03.65.Ud, 03.65.Yz, 03.65.Ta}
\maketitle

\section{Introduction}
\label{intro}
The ultra-precise determination of an optical phase~\cite{Holland:93, *Hwang:02, *Walther:04, *Mitchell:04}, 
the strength of a local magnetic field~\cite{Fleischhauer:00, *Budker:07, *Aiello:13}, atomic 
frequency~\cite{Wineland:92, *Bollinger:96, *Hempel:13}, spectroscopy~\cite{Roos:06, *Chwalla:07}, and 
clock-synchronization~\cite{Valencia:04, *deburgh:05} are just a few of the many celebrated achievements 
of quantum metrology~\cite{GLM06, *Giovanetti:11}.  If $N$ systems, deployed to probe the 
dynamics of a physical process, are prepared in a particular entangled state the precision in the 
estimation of the relevant parameters scales at the \defn{Heisenberg limit} $\mathcal{O}(N^{-1})$.
Contradistinctively,  if the $N$ probes are prepared in a product state then the \defn{standard limit}, $\mathcal{O}
(N^{-1/2})$, in precision is achieved~\cite{GLM04}.   Note that $\mathcal{O}(N^{-1})$ is the maximum achievable limit 
in precision---the Heisenberg limit---in cases where the spectral radius of the operator describing the unitary dynamics 
of the $N$ probes scales linearly with $N$, e.g., local and $k$-local Hamiltonians.  In what follows we consider cases where this latter requirement holds.

In quantum metrological scenarios considered most often the dynamics, in the absence of 
decoherence, describing the evolution of the $N$ probes are represented by the unitary operator $U(\theta)
=e^{i\theta H}$ where $\theta$ is the parameter, or set of parameters, to be estimated~\footnote{For the remainder of 
this paper we shall assume that the goal is to estimate a single parameter.} and $H$ is the 
Hamiltonian governing the evolution.  In most applications of quantum metrology to date the 
Hamiltonian is assumed to be \defn{local}, i.e.~$H=\sum_{i=1}^N h^{(i)}$, where $h^{(i)}$ acts on probe $i$.  
Metrological scenarios making use of one-dimensional cluster state Hamiltonians have also been studied~
\cite{RJ09}.  For local Hamiltonians employing qubits as probes the optimal state is the GHZ state~\cite{GHZ} in the 
case of atomic frequency spectroscopy~\cite{Wineland:92, Bollinger:96}, and the 
so-called NOON state in the case of optical interferometry~\cite{Hwang:02}.  The latter is a linear 
superposition of $N$ photons in an optical interferometer, with the $N$ photons being either in the upper or lower arm 
of the interferometer.

However, it is known that in the presence of local dephasing noise, where the noise operators commute with the 
Hamiltonian, GHZ and NOON states perform no better than separable states~\cite{Huelga:97}.   Indeed, in this 
scenario quantum metrology offers only a constant factor improvement, in the asymptotic limit, over the standard
limit~\cite{Escher:11, Kolodynski:12, *Kolodynski:13, Alipour:14}. The state that achieves this improvement, in the limit of large 
number of probes, is the so-called spin-squeezed state~\cite{OK01}.   

In this work we are mainly concerned with quantum metrology beyond the local Hamiltonian condition and the 
construction of states that perform favorably both in the presence and absence of noise.  Specifically, we use Lie algebraic techniques to:
\begin{enumerate}
\item Identify a class of Hamiltonians which belong to the Lie algebra of the special unitary group of two dimensions, 
$\mathfrak{su}(2)$.
\item For all such Hamiltonians, we provide a recipe for constructing states that achieve Heisenberg scaling in 
precision for noiseless metrology.
\item Determine the performance of these states for metrology in the presence of a local and nearest-neighbor 
Hamiltonian in the presence of local dephasing noise.
\end{enumerate}

Whereas there are other methods for identifying states that yield Heisenberg scaling precision, such as GHZ-type states
(see Sec.~\ref{sec:ParEst}),  our method for constructing alternative states for parameter estimation is of interest for the 
following two reasons.  On the one hand, the optimal states for parameter estimation are rather difficult to obtain in certain 
experimental set-ups, and on the other hand these optimal states are known to be extremely susceptible to noise, so 
much so that there precision scaling quickly deteriorates.  The pretty good states we introduce in this work, as well as 
their construction, maybe experimentally more accessible than the optimal states. More importantly these states, as we 
show, may perform better in the presence of noise than the optimal  states of noiseless metrology thereby paving the way 
towards practical quantum metrology in the presence of noise and imperfections.

Similar techniques, based on the Lie algebra $\mathfrak{su}(2)$, were 
used by Yurke and McCall to analyze the performance of optical interferometers in the presence of a local 
Hamiltonian~\cite{Yurke:86}.  Our results are broader as we develop techniques for 
constructing pretty good states for a large class of Hamiltonians including Hamiltonians based on 
graphs~\cite{Briegel:01, RJ09} as well as nearest-neighbor Hamiltonians which appear frequently in the study of 
interacting systems.  

Using this Lie algebraic approach we provide a recipe for constructing states that, 
in the absence of noise, achieve Heisenberg scaling in precision.  For the case of local Hamiltonians the states 
constructed using our procedure are the well-known Dicke states~\cite{Dicke:54}.  These are permutational symmetric 
states of $N$ particles with $0\leq k\leq N$ of the particles 
in the ground state and $N-k$ in the excited state. For other Hamiltonians in $\mathfrak{su}(2)$ for which permutation 
symmetry does not hold, such as nearest-neighbor Hamiltonians, our procedure can nevertheless  unambiguously 
give rise to states that achieve Heisenberg scaling for noiseless metrology.  

Finally, we determine the performance of the states constructed by our procedure in the presence of local dephasing 
noise.  Specifically, we use a well-known upper bound to the quantum Fisher information (QFI)~\cite{Escher:11} to 
determine the performance of our states for the case of local and nearest-neighbor Hamiltonians.  
For local Hamiltonians we find that states constructed via our method yield the same asymptotic bound as any 
state that is not a product state.  For finite $N$ our states yield a different pre-factor compared to the bound obtained 
for the product and GHZ states.  However, for nearest-neighbor 
Hamiltonians, we show that for any state the bound for the QFI of~\cite{Escher:11}, computed using a restricted set of Kraus operators, is equal to the QFI in the 
absence of noise. Consequently, we numerically compute the actual QFI for moderate 
values of $N$ and find that, for metrology using nearest-neighbor Hamiltonians under local dephasing noise, states 
constructed via our procedure are sub-optimal but outperform product states.  

The outline of this paper is as follows.  In Sec.~\ref{sec:Preliminaries} we review the mathematical background of 
both classical and quantum metrology (Sec.~\ref{sec:ParEst} and Sec.~\ref{sec:bounds} respectively), and Lie 
algebras (Sec.~\ref{sec:LieAlg}).  
In Sec.~\ref{sec:LieFormulation} we formulate noiseless quantum metrology in a Lie algebraic 
framework.  Using this framework we construct states that exhibit Heisenberg 
scaling (Sec.~\ref{sec:PGS}), and determine an entire class of Hamiltonians for which our construction 
applies (Sec.~\ref{sec:Construction}).  We illustrate our construction using a local 
Hamiltonian, as well as two non-local Hamiltonians. In Sec.~\ref{sec:noise} we 
study the performance of our constructed states  for noiseless metrology in the presence of local dephasing noise, and 
in  particular we provide both analytic (Sec.~\ref{sec:analytic}) and numerical results (Sec.~\ref{sec:numeric}) of the 
performance of these states for local, as well as nearest-neighbor Hamiltonians. We summarize and conclude in 
Sec.~\ref{sec:conclusion}.

\section{Preliminaries}
\label{sec:Preliminaries}
In this section we provide a brief background of both noiseless (Sec.~\ref{sec:ParEst}) and noisy quantum metrology 
(Sec.~\ref{sec:bounds}).  We outline key results in both these scenarios and introduce a characterization of pretty 
good states for noiseless and noisy metrology. For the sake of completeness we re-derive the bound by 
Escher {\it et al.}~\cite{Escher:11} pertaining to the best possible precision achievable by a quantum strategy in the 
presence of noise (Sec.~\ref{sec:bounds}).  
In Sec.~\ref{sec:LieAlg} we review the theory of $\mathfrak{su}(2)$ Lie algebras. 

\subsection{Classical and Quantum Parameter Estimation}
\label{sec:ParEst}

In a metrological scenario the goal is to estimate a parameter, $\theta\in\mathbb{R}$, of a population 
described by random variable, $X$, whose elements, $x\in\mathbb{R}$, correspond to measurement outcomes with 
respective probability distribution, $p(x|\theta)$, given a random finite  sample of $n$ data drawn from this 
population.  Using a suitable function, $\hat{\theta}:X^{n}\to\mathbb{R}$ (known as an  \defn{estimator}) an 
\defn{estimate} of $\theta$ is given by $\hat{\theta}\left(\{x_i\}_{i=1}^{n}\right)$.  

Two desired properties of any good estimator is unbiasedness and minimum variance.  An estimator is 
\defn{unbiased} if its expected value, $\langle\hat{\theta}\rangle$, with respect to the probability 
distribution $p(x|\theta)$ is equal to $\theta$. Furthermore, an estimator is said to have \defn{minimum variance}, 
$\delta\theta^2\equiv\langle(\frac{\hat{\theta}}{\mathrm{d}\langle\hat{\theta}\rangle/\mathrm{d}\theta}-\theta)^2\rangle$, if the 
variance of any other estimator is greater or equal to 
$\delta\theta^2$~\footnote{The factor $\mathrm{d}\hat{\theta}/\mathrm{d}\theta$ takes care of any difference in units 
between $\hat{\theta}$ and $\theta$.}.  A lower bound on the variance of any estimator is given by the 
well-known Cram\'{e}r-Rao inequality~\cite{Cramer:61}, 
\begin{equation}
 \delta\theta^2\geq\frac{1}{\nu \Phi(\theta)},
 \label{eq:CCR}
\end{equation}
where $\Phi(\theta)$ is the \defn{Fisher information} given by~\cite{Fisher:22}
\begin{equation}
 \Phi(\theta)=\int\frac{1}{p(x|\theta)}\left(\frac{\partial \ln p(x|\theta)}{\partial\theta}\right)^2\mathrm{d}x,
 \label{eq:CFI}
\end{equation}
and $\nu$ is the number of repetitions of the experiment. It is known that the lower bound in Eq.~\eqref{eq:CCR} is 
saturated in the limit $\nu\to\infty$ by the maximum likelihood estimator~\cite{Fisher:25}. The precision of the estimator 
is given by the square root of its variance.

The Fisher information quantifies the amount of information carried by the random variable $X$ about $\theta\in
\mathbb{R}$.  In quantum mechanics the parameter $\theta$ is imprinted in a state $\rho(\theta)\in\cB(\cH)$, where 
$\cB(\cH)$ denotes the set of bounded operators on the Hilbert space, $\cH$,  
of a quantum system after undergoing some evolution, as will be explained shortly, and the probability distribution is 
given by  
$p(x|\theta)=\mathrm{Tr}(M_x\rho(\theta)M_x^\dagger)$, where the set of measurement operators $\{M_x:\cB(\cH)\to\cB(\cH)\}$ 
satisfy $\sum_x M_x^\dagger M_x=I$.  Consequently, the Fisher information is different for different choice of 
measurement operators. Denoting by $\Phi_{M_x}(\theta)$ the Fisher information, Eq.~\eqref{eq:CFI}, for the 
measurement given by $\{M_x\}$, the \defn{quantum Fisher information} (QFI) is defined as
\begin{equation}
 \mathcal{F}(\rho(\theta)):=\max_{\{M_x\}}\Phi_{M_x}(\theta),
 \label{eq:QFI}
\end{equation}
and quantifies the amount of information about $\theta$ one learns when using the most informative 
measurement.  Consequently, the \defn{quantum Cram\'{e}r-Rao inequality} is given by~\cite{H76, *H80}
\begin{equation}
 \delta\theta^2\geq\frac{1}{\nu\mathcal{F}(\rho(\theta))}.
 \label{eq:QCR}
\end{equation} 
 
It has been shown in~\cite{BC94} that the QFI is given by
\begin{equation}
\mathcal{F}(\rho(\theta))=\mathrm{Tr}\left[L_{\theta}\rho(\theta)L_{\theta}\right],
\label{eq:QFI1}
\end{equation}
with
\begin{equation}
L_{\theta}=2\sum_{\alpha,\beta}\frac{\bra{\alpha}\dot{\rho}(\theta)\ket{\beta}}{\lambda_\alpha+\lambda_
\beta}\ket{\alpha}\bra{\beta}
\label{eq:SLD}
\end{equation}
the \defn{symmetric logarithmic derivative}, where $\dot{\rho}(\theta)=\partial\rho(\theta)/\partial\theta$, $
\lambda_\alpha$ the eigenvalues of $\rho
(\theta)$, $\ket{\alpha}$ the corresponding eigenvectors, and the sum in Eq.~\eqref{eq:SLD} is over all $
\alpha,\beta$ satisfying $\lambda_\alpha+\lambda_\beta\neq0$. The most informative 
measurement is the one whose elements are the projectors on the eigenspaces of the symmetric logarithmic 
derivative. 

Two important properties of the QFI are its additivity, $\mathcal{F}\left(\rho(\theta))^{\otimes 
N}\right)=N\mathcal{F}\left(\rho(\theta)\right)$, and convexity, 
$\mathcal{F}\left(\sum_i p_i \rho_i(\theta)\right)\leq\sum_i p_i\mathcal{F}
\left(\rho_i(\theta)\right)$~\cite{BCM96}.

The parameter $\theta$ is imprinted in the state, $\rho\in\cB(\cH)$, of a quantum system by a completely 
positive, trace-preserving (CPT) map, $\cE_\theta:\cB(\cH)\to\cB(\cH),\, \rho(\theta)=\cE_\theta\left(\rho\right)$. 
For different values of $\theta$, $\rho(\theta)$ traces a curve in the space of bounded operators. 
\begin{figure}[htbp]  
\includegraphics[keepaspectratio, width=7cm]{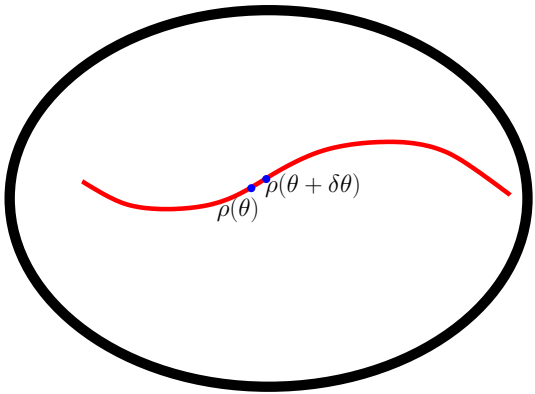}  
\caption{The curve traced by the set of states $\cE_\theta\left(\rho\right)$.  The steeper the gradient of the curve 
between $\theta$ and $\theta+\delta\theta$ the more distinguishable the states $\cE_\theta\left(\rho\right)$ and $
\cE_{\theta+\delta\theta}\left(\rho\right)$ become.}
\label{fig:1}
\end{figure}
Determining the value of $\theta$ is equivalent to distinguishing between $\rho(\theta)$ and 
$\rho(\theta+\delta\theta)$ (see Fig.~\ref{fig:1}). For the case where $\rho$ is 
pure, i.e.~$\rho=\ketbra{\psi}{\psi}$, and in the absence of noise, i.e.~$\cE_{\theta}=e^{i\theta H}\rho e^{-i\theta H}$, 
with $H$ the generator of shifts in $\theta$ (the Hamiltonian), $\mathcal{F}\left(\rho(\theta)\right)=4(\Delta H)^2$, where 
$(\Delta H)^2\equiv\langle H^2\rangle-\langle H\rangle^2$
is the variance of $H$~\cite{BCM96}.  

It follows that in order to obtain the best estimate of $\theta$ in the absence of noise one must 
use an initial pure state, $\ket{\psi}$, that has the largest variance with respect to $H$. It can be shown 
that for any Hamiltonian, $H$, acting on $N$ quantum systems $(\Delta H)^2$  is optimized by states of the 
form~\cite{GLM06, *Giovanetti:11,Maccone:13}
\begin{equation}
 \ket{\psi}=\sqrt{\frac{1}{2}}\left(\ket{\Lambda_{\mathrm{min}}}+e^{i\phi}\ket{\Lambda_{\mathrm{max}}}\right), 
 \label{eq:opt-state}
\end{equation}
with $\phi\in(0,2\pi]$ arbitrary, and where $\ket{\Lambda_{\mathrm{min}(\mathrm{max})}}$ is 
the eigenstate of $H$ corresponding to the minimum (maximum) eigenvalue.  

For the case of local Hamiltonians such as $H=\sum_i \sz^{(i)}$, where $\sz^{(i)}$ is the Hamiltonian acting on the 
$i^{\mathrm{th}}$ system, the states in Eq.~\eqref{eq:opt-state} correspond to GHZ (NOON)-like states in frequency 
spectroscopy and optical interferometry respectively, and achieve a precision 
\begin{equation}
\delta\theta^2\geq\frac{1}{\nu N^2},
\label{eq:HL}
\end{equation}
known as the Heisenberg limit~\cite{GLM04}.  We refer to these states as \defn{optimal states} for noiseless 
metrology.  Note that for the latter one can achieve the Heisenberg limit by 
employing a single system that undergoes $N$ sequential applications of the unitary 
$e^{i\theta \sz}$~\cite{GLM06,Rudolph:03, HBBWP07,*HBBMWP09, Maccone:13}.  
For local Hamiltonians the $N$-partite entangled 
states of Eq.~\eqref{eq:opt-state} offer a quadratic improvement, with respect to the number of probes used, over the 
best strategy employing separable states of $N$ probes achieved by $\ket{\phi_{\mathrm{PS}}}=\ket{+}^{\otimes N}$~\cite{GLM06,*Giovanetti:11}. Indeed, as the variance of any operator scales at most quadratically with it's spectral radius, the latter scaling linear with $N$ implies that the ultimate achievable precision limit scales at most inversely proportional to $N$, i.e., the Heisenberg limit.  

From the preceding discussion it is clear that a state is good for noiseless parameter estimation if the variance 
of the Hamiltonian with respect to this state scales as  $\mathcal{O}(N^\alpha)$ for $1<\alpha\leq 2$.  Hence, 
a hierarchy of resources for noiseless parameter estimation can be established with the optimal states, 
Eq.~\eqref{eq:opt-state}, being the ultimate resources~\footnote{In the presence of noise all sates for which $(\Delta H)^2=\mathcal{O}(N)$ are potentially interesting.}. Any state for which the variance of the Hamiltonian is 
$\mathcal{O}(N^2)$, i.e.~achieves Heisenberg scaling up to a constant factor independent of $N$, is 
desirable for noiseless parameter estimation.  Identifying such states may be important in cases where the 
ultimate resources are unavailable, or highly costly to prepare.  

For example, let $H=\sum_k \lambda_k\ketbra{\lambda_k}{\lambda_k}$ be the spectral decomposition of $H$ with spectral radius, 
$\varrho(H)=\mathcal{O}(N)$, and consider an arbitrary state 
\begin{equation}
\ket{\psi}=\sum_{k} c_k\ket{\lambda_k},
\label{eq:PG-state}
\end{equation}
where $c_k\in\mathbb{C}$. The variance of $H$ with respect to the state of Eq.~\eqref{eq:PG-state} can be trivially 
computed to be 
\begin{equation}
(\Delta H)^2=\sum_k |c_k|^2\lambda_k^2-\left(\sum_k |c_k|^2\lambda_k\right)^2.
\label{eq:PG-variance}
\end{equation}
Viewing the state as a classical probability distribution over the eigenvalues of $H$, we say a state is \defn{pretty 
good} for noiseless metrology if its probability distribution, weighted by the eigenvalues of $H$, has a variance 
of $\mathcal{O}(N^2)$.  For the remainder of this work we mainly restrict to Hamiltonians with a 
\defn{homogeneously gapped spectrum}, i.e.~where the ordered eigenvalues, $\lambda_k$, of $H$ satisfy 
$|\lambda_{k+1}-\lambda_k|=c\in\mathbb{R}, \; \forall k$.  A few examples of pretty good states for such Hamiltonians 
are shown in Fig.~\ref{fig:2}.  
\begin{figure}[ht]
\includegraphics[width=8cm]{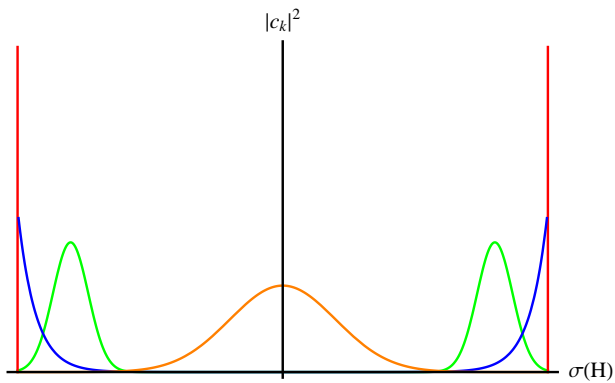}
\caption{Probability distributions of the coefficients, $|c_k|^2$, of the arbitrary state in Eq.~\eqref{eq:PG-state} over a 
Hamiltonian with equally gapped spectrum $\sigma(H)$. The horizontal axis represents the continuous limit of the discrete spectrum of $H$.  The red distribution gives the optimal variance 
$(\Delta H)^2=N^2$, whereas the orange distribution gives $(\Delta H)^2=N$.  The blue and green 
distributions have variance $(\Delta H)^2=\mathcal{O}(N^2)$ and correspond to pretty good states.}
 \label{fig:2}
\end{figure}

\subsection{Bounds for noisy metrology}
\label{sec:bounds}
As the QFI is convex it follows that in noiseless metrology no increase in 
precision is gained by preparing the $N$ probes in a mixed state. However, in the presence of 
decoherence the evolution of the $N$ probes is described by the CPT map, $\cE_{\theta}$, whose  Kraus 
decomposition contains more than a single Kraus operator~\cite{Kraus:83}.  Consequently,  
the state $\cE_\theta\left(\rho\right)\in\cB(\cH)$, where $\rho$ is the initial state of the $N$ probes,  will in general be a 
mixed state.   As Eq.~\eqref{eq:SLD} requires both the eigenvalues and eigenvectors of $\cE_\theta\left(\rho\right)$, 
computing the QFI for large $N$ becomes intractable. Due to this fact, much of the theoretical developments in noisy 
quantum metrology have focused on deriving upper bounds for the QFI.   

A promising route for placing an upper bound on the QFI utilizes channel 
extension~\cite{Fujiwara:08, Kolodynski:12,*Kolodynski:13} and channel purification based 
techniques~\cite{Escher:11, Escher:12} which we re-derive here for the 
sake of completeness. One can write the state of the $N$ probes after the noisy evolution as 
\begin{align}\nonumber
\rho(\theta)&=\mathrm{tr}_E\left[\tilde{U}^{(SE)}_\theta[\ketbra{\psi_S}{\psi_S}
\otimes(\ketbra{0}{0})_E]\tilde{U}^{(SE)\dagger}_\theta\right]\\
&\equiv\mathrm{tr}_E\left[\ketbra{\Psi(\theta)_{(SE)}}{\Psi(\theta)_{(SE)}}\right],
\label{eq:purification}
\end{align}
where the subscripts $S\, (E)$ refer to the system (environment) respectively, $\tilde{U}^{(SE)}_\theta$ is a 
unitary acting on both the system and the environment, and 
$\ket{\Psi(\theta)_{(SE)}}\equiv \tilde{U}^{(SE)}_\theta(\ket{\psi_S}\otimes\ket{0_E})$ is a purification of $\rho(\theta)$. 
Due to the partial trace over the environment in Eq.~\eqref{eq:purification}, the purification, $\ket{\Psi(\theta)_{(SE)}}$, of 
$\rho(\theta)$ is unique up to an isometry, $V^{(E)}_\theta$, acting on the $E$.
As one can gain more information about $\theta$ by measuring both $S$ and $E$ the QFI of $\ket{\Psi(\theta)_{(SE)}}$ provides an upper bound on the QFI of $\rho(\theta)$
\begin{align}\nonumber
\mathcal{F}\left(\rho(\theta)\right)\leq&\mathcal{F}\left(\ketbra{\Psi(\theta)_{(SE)}}{\Psi(\theta)_{(SE)}}\right)\\
&\equiv C_Q\left(\ketbra{\psi_S}{\psi_S},\,K_m(\theta)\right),
\label{eq:bound}
\end{align}
where $K_m(\theta)\equiv_E\bra{m}\tilde{U}^{(SE)}_\theta\ket{0}_E$ are the Kraus operators describing the action of 
the CPT map $\cE_\theta$.

The upper bound in Eq.~\eqref{eq:bound} involves determining the QFI of a pure state which, 
using Eq.~\eqref{eq:SLD}, can be easily determined to be~\cite{BC94} 
\begin{align}\nonumber
\mathcal{F}\left(\ketbra{\Psi(\theta)_{(SE)}}{\Psi(\theta)_{(SE)}}\right)&=4\left(\bracket{\Psi(\theta)_{(SE)}^\prime}
{\Psi(\theta)_{(SE)}^\prime}-\right.\\
&\left.\left|\bracket{\Psi(\theta)_{(SE)}^\prime}{\Psi(\theta)_{(SE)}}\right|^2\right),
\label{eq:QFI-purestate}
\end{align}
where $\ket{\Psi(\theta)_{(SE)}^\prime}\equiv\frac{\mathrm{d}\ket{\Psi(\theta)_{(SE)}}}{\mathrm{d}\theta}$. In terms of 
the Kraus operators, $K_m(\theta)$, Eq.~\eqref{eq:bound} reads~\cite{Escher:11} 
\begin{equation}
C_Q(\ketbra{\psi_S}{\psi_S},\,K_m(\theta))=4\left(\bra{\psi_S}A_1\ket{\psi_S}-
\left(\bra{\psi_S}A_2\ket{\psi_S}\right)^2\right),
\label{eq:escher-bound}
\end{equation}
where 
\begin{align}\nonumber
A_1&=\sum_m \frac{\mathrm{d}K_m(\theta)^\dagger}{\mathrm{d}\theta}\frac{\mathrm{d}K_m(\theta)}{\mathrm{d}\theta}
\\
A_2&=i\sum_m\frac{\mathrm{d}K_m(\theta)^\dagger}{\mathrm{d}\theta}K_m(\theta).
\label{eq:escher-operators}
\end{align}

That the upper bound of Eq.~\eqref{eq:bound} is attainable can be seen by noting that the Bures fidelity, 
$F\left(\rho(\theta),\rho(\theta+\delta\theta)\right)^2\equiv\left(\mathrm{tr}\sqrt{\rho(\theta)^{1/2}\rho(\theta+\delta\theta)\rho(\theta)^{1/2}}\right)^2$, between two adjacent density operators 
$\rho(\theta)$, $\rho(\theta+\delta\theta)$ (see Fig.~\ref{fig:1}) can be expanded to second order in $\delta\theta$ as 
\begin{equation}
F\left(\rho(\theta),\rho(\theta+\delta\theta)\right)^2=1-\frac{\delta\theta^2}{4}\,\mathcal{F}\left(\rho(\theta)\right)
+\mathcal{O}(\delta\theta^3).
\label{eq:bures-fidelity}
\end{equation}
In addition, Uhlmann's theorem states that $F\left(\rho(\theta),\rho(\theta+\delta\theta)\right)^2
=\max_{\left\{\ket{\Psi(\theta+\delta\theta)_{(SE)}}\right\}}|\bracket{\Phi(\theta)_{(SE)}}{\Psi(\theta+\delta\theta)_{(SE)}}|
^2$, where $\ket{\Phi(\theta)_{(SE)}}$ is a purification of $\rho(\theta)$, $\ket{\Psi(\theta+\delta\theta)_{(SE)}}$ a 
purification of $\rho(\theta+\delta\theta)$, and the maximization is over all purifications of $\rho(\theta+\delta\theta)$~
\cite{MikeIke}. Performing a Taylor expansion of $|\bracket{\Phi(\theta)_{(SE)}}{\Psi(\theta+\delta
\theta)_{(SE)}}|^2$ up to second order in $\delta\theta$ yields  
\begin{align}\nonumber
&|\bracket{\Phi(\theta)_{(SE)}}{\Psi(\theta+\delta\theta)_{(SE)}}|^2=1-\delta\theta^2\times\\
&\left(\bracket{\Psi(\theta)_{(SE)}^\prime}{\Psi(\theta)_{(SE)}^\prime}-\left|\bracket{\Psi(\theta)_{(SE)}^\prime}
{\Psi(\theta)_{(SE)}}\right|^2\right) +\mathcal{O}(\delta\theta^3).
\label{eq:fidelity}
\end{align}
Hence, up to second order in $\delta\theta$, the maximization over $\ket{\Psi(\theta+\delta\theta)_{(SE)}}$ required to 
compute the fidelity amounts to the minimization of the second term in Eq.~\eqref{eq:fidelity} over 
all $\ket{\Psi(\theta+\delta\theta)_{(SE)}}$.  From Eq.~\eqref{eq:fidelity} and Eq.~\eqref{eq:bures-fidelity} it follows 
that  
\begin{align}\nonumber
\mathcal{F}(\rho(\theta))&=4\min_{\left\{\ket{\Psi(\theta)_{(SE)}}\right\}}\left(\bracket{\Psi(\theta)_{(SE)}^\prime}
{\Psi(\theta)_{(SE)}^\prime}-\right.\\ \nonumber
&\left.\left|\bracket{\Psi(\theta)_{(SE)}^\prime}{\Psi(\theta)_{(SE)}}\right|^2\right)\\
&=\min_{\{K_m(\theta)\}}C_Q(\ketbra{\psi_S}{\psi_S},\,K_m(\theta)).
\label{eq:attainable}
\end{align}
Hence, the bound of Eq.~\eqref{eq:bound} is attained by minimizing over all possible Kraus decompositions of the 
CPT map, $\cE_\theta$, which is equivalent to optimizing over the unitary transformation $V^{(E)}_\theta$ on the 
environment in Eq.~\eqref{eq:purification}.

The above technique has been applied to quantum metrology in the presence of local, 
uncorrelated noise, such as dephasing, loss, and noise of full 
rank~\cite{Escher:11,Kolodynski:12,*Kolodynski:13}~\footnote{A full rank channel is a channel that lies in the interior of 
the space of quantum channels.}.  In the case of local unitary evolution and local uncorrelated 
dephasing noise, where the dephasing operators commute with the Hamiltonian, it was shown that the 
ultimate precision achievable using a GHZ or NOON state is given by
\begin{equation}
 \delta\theta^2\geq\frac{1}{\nu\eta^{2N}N^2},
 \label{eq:bound-dephasing}
\end{equation}
where $0\leq\eta<1$ denotes the strength of dephasing with $\eta=1$ meaning no dephasing at 
all~\cite{Kolodynski:13}.  Whereas for small $N$ GHZ and NOON states still exhibit precision inversely proportional to 
$N^2$, for large $N$ precision in phase estimation quickly decreases below the standard limit.   

The usefulness of optimal states completely disappears if one 
wishes to estimate frequency rather than phase.  In this case $\theta=\omega t$ and $\omega$ is the parameter to be estimated and  the resources are the 
number of probes used and the \defn{total running time} for the experiment, $T=\nu t$, where $t$ is the time for a 
single experimental  run and  $\nu$ are the number of repetitions.  Here, one not only needs to optimize over the 
measurements but also over the time, $t$, that these measurements need to be performed.  The variance in 
estimating frequency using a GHZ  state was shown to be~\cite{Huelga:97}
\begin{equation}
\delta\omega^2\geq\frac{2\gamma e}{NT},
\label{eq:bound-frequency-ps}
\end{equation}
with $\gamma$ the dephasing parameter.  Exactly the same precision is achieved if the $N$ probes are 
prepared in a separable state.  

However, for a particular measurement strategy commonly employed in Ramsey spectroscopy it was shown that~\cite{Huelga:97}
\begin{equation}
\delta\omega^2\geq\frac{2\gamma}{NT},
\label{eq:bound-frequency-sss}
\end{equation}
i.e.~only a factor of $e^{-1}$ improvement over the standard limit can be achieved.  Furthermore, the bound in 
Eq.~\eqref{eq:bound-frequency-sss} was shown to be asymptotically achievable by a spin-squeezed state 
with a particular squeezing parameter that decreases with $N$~\cite{OK01}. Note that this result provides a lower 
bound on the 
QFI as the measurement is fixed. Using the purification techniques discussed above Escher {\it et al.} derived an 
upper bound, equal to Eq.~\eqref{eq:bound-frequency-sss}, on the QFI proving that this is indeed the ultimate 
achievable precision~\cite{Escher:11}.  It is worth noting that recent work  proves that the use of quantum error-correcting codes can help suppress the decoherence effects and thus restore the Heisenberg limit in some noisy 
metrological scenarios~\cite{Dur:13,*Arad:13,*Kessler:13,*Ozeri:13}.  

\subsection{The $\mathfrak{su}(2)$ Lie algebra}
\label{sec:LieAlg}

In this sub-section we briefly review the $\mathfrak{su}(2)$ Lie algebra. This algebra is the familiar algebra for 
angular momentum in quantum mechanics.  Thus, we simply outline the key properties that will be useful for 
our purposes and refer the reader to~\cite{Sakurai:94} for further details.

Let $V$ be an $M$-dimensional vector space over the field $\mathfrak{F}$ equipped with an operation $[\cdot,
\cdot]:V\times V\to V$, the \defn{Lie bracket} or \defn{commutator}, and let $\{S_i\}_{i=1}^M\in V$ be a 
set of linearly independent vectors.  Then $\{S_i\}_{i=1}^M$ form a \defn{Lie algebra} if the 
following hold:
\begin{enumerate}
\item $[S_k,S_l]= \sum_m c_{klm} S_m$, where $c_{klm}\in \mathfrak{F}$ are the \defn{structure constants} of the 
algebra,
\item $[S_k,S_k]=0$ for all $S_k\in\{S_i\}_{i=1}^M$,
\item For any $S_k,S_l, S_m\in\{S_i\}_{i=1}^M,\:[S_k,[S_l,S_m]]+
[S_l,[S_m,S_k]]+[S_m,[S_k,S_l]]=0$.
\end{enumerate}

The $\mathfrak{su}(2)$ Lie algebra is a three-dimensional vector space whose basis elements, $\{S_i\}_{i=1}^3$, are 
the generators of the algebra. The  $\mathfrak{su}(2)$ structure constants are given by $c_{klm}=ic\epsilon_{klm}$, 
where $c\in\mathbb{R}$ and $\epsilon_{klm}$ is the Levi-Civita symbol.  For $c=0$ one obtains a trivial, 
three-dimensional, abelian algebra. For $c=2$ one possible set of generators for $\mathfrak{su}(2)$ are the Pauli 
matrices $\{\sx,\,\sy,\,\sz\}$.

For $\mathfrak{su}(2)$ the operator $\bm J^2=\sum_i S_i^2$, known as the \defn{Casimir invariant} 
of $\mathfrak{su}(2)$, obeys $[\bm J^2,A]=0,\: \forall A\in\mathfrak{su}(2)$.  If the elements 
$\{S_i\}_{i=1}^3$ represent Hermitian operators, acting on the Hilbert space of a quantum system, 
then a convenient basis for the Hilbert space is $\{\ket{j,m,\beta}\}$, where $j$ and $m$ label the eigenvalues of $\bm 
J^2$ and $S_z$ (which in our notation is denoted by $S_3$) respectively, and $\beta$ is a multiplicity label indicating 
the degeneracy of a given pair of labels $(j,m)$.  Then, 
\begin{align}\nonumber
 \bm J^2\ket{j,m,\beta}&=c^2\, j(j+1)\ket{j,m,\beta}\\
 S_3\ket{j,m,\beta}&=c\, m\ket{j,m,\beta},
 \label{eq:angular-momentum}
\end{align}
where $j$ is either an integer or half-odd integer, and $m$ can take any of the $2j+1$ values in the 
interval $-j\leq m\leq j$.  For a given $j$ one obtains all values of $m$ by starting from $m=\pm j$ and 
repeatedly applying the \defn{ladder operators}
\begin{equation}
 J^{(3)}_\pm:=\frac{S_1\pm iS_2}{\sqrt{2}},
 \label{eq:ladders}
\end{equation}
respectively. Here, and in the following, the notation $J^{(k)}_\pm$ denotes the ladder operators that raise (lower) the 
eigenstates of $S_k$ and are defined as $J^{(k)}_\pm\equiv\frac{(S_l\pm S_m)}{\sqrt{2}}$, where the indices 
$(klm)$ are cyclic permutations of $(123)$.  We remark that the set 
$\{S_k, J^{(k)}_\pm\}$ constitute another set of generators of $\mathfrak{su}(2)$ with
\begin{equation}
[S_k,J^{(k)}_\pm]=\pm c J^{(k)}_\pm, \:\: [J^{(k)}_+,J^{(k)}_-]=S_k.
\end{equation}
Note that $\bm J^2$ can also be written as $\bm J^2=S_k^2+\{J^{(k)}_+,J^{(k)}_-\}$, where $\{A,B\}=AB+BA$ 
is the \defn{anti-commutator}.  

We now show how noiseless parameter estimation, where the Hamiltonian is one of the generators of $\mathfrak{su}(2)$, can be re-phrased in purely Lie algebraic terms, and derive the pretty good states for noiseless quantum metrology. In Sec.~\ref{sec:noise} we investigate the performance of these pretty good states in the presence of noise.

\section{Lie algebraic formulation of noiseless quantum metrology}
\label{sec:LieFormulation}

In this section we formulate noiseless quantum metrology using the $\mathfrak{su}(2)$ Lie algebra, and provide a 
recipe for constructing pretty good states (Sec.~\ref{sec:PGS}).  In addition,  we derive a class of Hamiltonians for 
which our construction can be applied (Sec.~\ref{sec:Construction}). Specifically, we consider a unitary evolution given 
by $e^{i\theta H}$, and restrict ourselves to homogeneously gapped Hamiltonians.  

As shown in Sec.~\ref{sec:ParEst} the quantum Cram\'{e}r-Rao bound in this case is given by 
$\delta\theta^2\geq\frac{1}{4\nu(\Delta H)^2}$. The states that optimize the variance are linear superpositions of states corresponding to the maximum and minimum eigenvalues of $H$.  However, such states may be difficult to generate. Expressing the variance of $H$  as a function of the generators of 
$\mathfrak{su}(2)$, and using algebraic techniques, we construct states whose variance also scales quadratically with $N$.  Such states could be easier to prepare than the optimal states yet would still yield Heisenberg scaling in precision.  
For ease of exposition we simply state the important results in this section, and defer all proofs to 
Appendices~\ref{append2} and~\ref{append3}.

\subsection{Pretty good states for noiseless metrology}
\label{sec:PGS}

Let us assume that $H\equiv S_1\in\mathfrak{su}(2)$, and that there exist two more operators, $S_2,\,S_3$, such that 
$\{S_1,\,S_2,\,S_3\}$ are generators for $\mathfrak{su}(2)$.  In Sec.~\ref{sec:Construction} we will establish necessary 
conditions  for $\{H,\,S_2,\, S_3\}$ to be the generators of $\mathfrak{su}(2)$ for Hamiltonians with homogeneously 
gapped spectra.  Many of the Hamiltonians considered in the context of quantum metrology thus far are 
homogeneously gapped and thus satisfy these conditions (see Sec.~\ref{sec:examples}). 

The following theorem shows how to construct pretty good states for noiseless metrology. Starting from the ground 
state of one of the generators different from $H$ (say $S_3$), one simply applies the raising operator $J^{(3)}_+$ (see 
Eq.~\eqref{eq:ladders}) a sufficient number of times.

\begin{theorem}
Let $\{S_1,\,S_2,\,S_3\}$ be a set of generators for $\mathfrak{su}(2)$. Assume, without loss of generality, 
that $H\equiv S_1$, and let $\ket{\psi_\mathrm{min}}$ be an eigenstate of $S_3$ corresponding to the smallest 
eigenvalue.  Then the variance of $H$ with respect to the state 
\begin{equation}
\ket{\psi}=\sqrt{\frac{1}{\mathcal{N}}} J_+^{(3)k}\ket{\psi_\mathrm{min}},
\label{eq:thm1}
\end{equation} 
where $\mathcal{N}$ denotes the normalization constant, scales as half the 
\defn{spectral radius} of $\bm J^2$, $\varrho(\bm J^2)$, if $k=\lceil\frac{2j_\mathrm{max}+1}{2}\rceil$, where $j_
\mathrm{max}$ is related to the maximum eigenvalue of $\bm J^2$ via Eq.~\eqref{eq:angular-momentum}, and $\lceil\cdot\rceil$ is the ceiling function.
\label{thm1}
\end{theorem}
The proof of Theorem~\ref{thm1} can be found in Appendix~\ref{append2}.

In order to achieve Heisenberg scaling we require that $\varrho(\bm J^2)=\mathcal{O}(N^2)$.  
As $\varrho(\bm J^2)\geq \varrho(S_1^2)$, it is sufficient that $\varrho(S_1)\propto N$. 
Recall that we consider here, as in all other realistic quantum metrology scenarios, that $H$ is the sum of $N$ 
$k$-local Hamiltonians, i.e., $H=\sum_{i=1}^{N-k} h^{(i)}$, with $h^{(i)}$ a $k$-nearest neighbour Hamiltonian where $k$ 
is independent of $N$.  
In this case $\varrho=\mathcal{O}(N)$, and the QFI scales as $\mathcal{O}(N^2)$.

For almost all Hamiltonians considered so far in quantum metrology, this is indeed the case.  

We now show that for the case where $H\equiv 1/2\, S_z=1/2\sum_i\sigma_z^{(i)}$,  the states of 
Theorem~\ref{thm1} are the well-known Dicke states in the $x$-basis~\cite{Dicke:54}. Indeed, the set of 
operators $\{S_x, S_y,S_z\}$, with $S_y$ defined 
similar to $S_z,\,S_x$, are the generators of the $\mathfrak{su}(2)$ Lie algebra. 
It is easy to show that $J^{(x)}_+=\frac{1}{\sqrt{2}}(S_y+iS_z)$, and that $\ket{\psi_{\mathrm{min}}}$ of $S_x$ is given 
by $\ket{-}^{\otimes N}$.  Applying $J^{(x)}_+$, $\lfloor \frac{N}{2}\rfloor$ times to the state $\ket{-}^{\otimes N}$ gives
\begin{align}\nonumber
\ket{\psi}&=\sqrt{\frac{1}{\binom{N}{\left\lfloor\frac{N}{2}\right\rfloor}}}\sum_{\pi\in S_N}\ket{-_{\pi(1)}}
\ket{+_{\pi(2)}}\ket{-_{\pi(3)}}\ldots\ket{+_{\pi(N)}}\\
&\equiv\ket{N,\left\lfloor\frac{N}{2}\right\rfloor}_x
\label{eq:dickestate}
\end{align}
where $\lfloor\cdot\rfloor$ denotes the floor function,  $\pi\in S_N$ denotes an element of the permutation 
group of $N$ objects, and the sum in Eq.~\eqref{eq:dickestate} runs over all permutations. The state in 
Eq.~\eqref{eq:dickestate} is the Dicke 
state of $N$ systems, $\lfloor \frac{N}{2}\rfloor$ of which are in the $+1$ eigenstate of $\sigma_x$ and the rest are in 
the $-1$ eigenstate. The variance of this state is given by $\frac{N/2(N/2+1)}{2}$ for $N$ even and 
$\frac{(N/2(N/2+1)-1/2)}{2}$ if $N$ is odd. Hence, the state $\ket{N,\left\lfloor\frac{N}{2}\right\rfloor}_x$ is a pretty good 
state.

Note that if one were to choose a different set of generators, say 
$W_2=\alpha S_2+\beta S_3, \, W_3=-\beta S_2+\alpha S_3$, then the pretty good states obtained with 
$W_2,\,W_3$ differ from those obtained from $S_2,\, S_3$ only in the relative phases of the coefficients of the state, 
expanded in the eigenbasis of $H$. This is due to the fact that the generators of $\mathfrak{su}(2)$ form a basis for a 
three-dimensional vector space, with the basis $\{H, W_2,\, W_3\}$ obtained from $\{H, S_2,\, S_3\}$ by an 
appropriate rotation about the vector corresponding to the generator $H$. 

As we discuss at the end of the next section, our procedure for constructing pretty good states also applies for more general Hamiltonians.

\subsection{Constructing $\mathfrak{su}(2)$ from the Hamiltonian}
\label{sec:Construction}

We now determine necessary conditions for a homogeneously gapped Hamiltonian to be an element of 
$\mathfrak{su}(2)$ and provide a prescription for how, given the Hamiltonian, one can construct the two remaining 
generators of $\mathfrak{su}(2)$.  At the end of this section we discuss how a similar construction can be applied to 
block diagonal Hamiltonians, where each of the blocks is homogeneously gapped.

Let $S_1$ be our Hamiltonian which, using the spectral decomposition, can be written as 
\begin{align}\nonumber
S_1&=\sum_{k=1}^n\sum_{i=1}^{d_k} \lambda_k\ketbra{k,i}{k,i}\\
&\equiv\sum_{k=1}^n \lambda_k\ketbra{k}{k}\otimes \one_{d_k},
\label{eq:spectraldecomp}
\end{align}
where $\lambda_k$ are the eigenvalues, $\ket{k,i}$ the corresponding eigenvectors, and $\{\ket{k,i}\}_{i=1}^{d_k}$ form an orthonormal basis of a subspace whose dimension, $d_k$, corresponds to the multiplicity of the 
$k^{\mathrm{th}}$ eigenvalue.  Without loss of generality we may order the eigenvalues of $S_1$ in decreasing 
order such that $\lambda_{k}>\lambda_{k+1},\,\forall k$ and shift the entire spectrum of $H$ such that $-\lambda_k=\lambda_{n-k+1}$.  We seek two Hermitian operators, $S_2,\,S_3$, such 
that 
\begin{align}\nonumber
[S_2,S_3]&=i c S_1\\  \nonumber
[S_3, S_1]&=i c S_2\\
[S_1,S_2]&=i c S_3,
\label{eq:Lie-algebra}
\end{align}
where $c>0$. The following lemma, whose proof can be found in Appendix~\ref{append3}, establishes the form the operators $S_2,\, S_3$ must take.
\begin{lemma}
Let $S_1$ be given as in Eq.~\eqref{eq:spectraldecomp} and let $S_2,\, S_3$ be two Hermitian operators. 
If the spectrum of $S_1$ is homogeneously gapped, i.e.~$|\lambda_{k+1}-\lambda_{k}|=c,\,\forall k$, and the 
conditions in Eq.~\eqref{eq:Lie-algebra} hold then
\begin{align}\nonumber
S_2&=\sum_{k=1}^n\ket{k+1}\bra{k}\otimes S_2^{(k,k+1)}+\ket{k}\bra{k+1}\otimes S_2^{(k+1,k)}\\
S_3&=\sum_{k=1}^n\ket{k+1}\bra{k}\otimes S_3^{(k,k+1)}+\ket{k}\bra{k+1}\otimes S_3^{(k+1,k)},
\end{align}
where $S_2^{(k,l)}$ ($S_3^{(k,l)}$) are $d_l\times d_k$ matrices, and $S_2^{(k,k+1)}=-i S_3^{(k,k+1)},\,\forall k$.
\label{lem1}
\end{lemma}

It now remains to determine the form of the matrices $S_2^{(k+1,k)}, S_3^{(k+1,k)}$.  The following theorem establishes necessary conditions on the multiplicities, $d_k$, of $S_1$ in order for $\{S_1,S_2, S_3\}$ to be the generators of $\mathfrak{su}(2)$.   In addition, Theorem~\ref{thm2} provides one possible solution for the operators $S_2, \,S_3$. The proof of Theorem~\ref{thm2} can be found in Appendix~\ref{append3}.

\begin{theorem}
Let $S_1$ be given by Eq.~\eqref{eq:spectraldecomp} with the eigenvalues of $S_1$ satisfying 
$\lambda_{k}-\lambda_{k+1}=c, \,\forall k\in(1,\ldots, n)$. 
In addition, let the operators $S_2,\, S_3$ be given as in Lemma~\ref{lem1}. Necessary conditions for 
Eq.~\eqref{eq:Lie-algebra} to hold are that $d_{k+1}\geq d_{k}$ for $1\leq k\leq\lfloor\frac{n}{2}\rfloor$, 
and  $d_{k}=d_{n+1-k}$.  Furthermore, one possible solution for the matrices $S_3^{(k,k+1)}$ is given by the $d_{k+1}\times d_k$ matrix 
\begin{align}\nonumber
S_3^{(k,k+1)}&=\sqrt{\frac{c}{2}}\,\mathrm{diag}\left(\underbrace{\sqrt{\sum_{i=1}^k\lambda_i}}_{d_1\,\mathrm{times}},\underbrace{\sqrt{\sum_{i=2}^k\lambda_i}}_{(d_2-d_1)\,\mathrm{times}},\ldots,\right.\\
& \left.\underbrace{\sqrt{\sum_{i=k-1}^k\lambda_i}}_{(d_{k-1}-d_{k-2})\,\mathrm{times}},\underbrace{\sqrt{\lambda_k}}_{(d_k-d_{k-1})\,\mathrm{times}}\right).
\label{eq:thm2}
\end{align}
\label{thm2}
\end{theorem} 

The proof uses the fact that $[S_2,S_3]=i\, c\, H$, and the form of $S_2$ and $S_3$ given in 
Lemma~\ref{lem1} to establish a set of $n$ equations involving $n$ unknown operators. In order for the 
system of $n$ equations to be solvable, it is necessary that $d_{k+1}\geq d_k$ for 
$1\leq k\leq\lfloor\frac{n}{2}\rfloor$ and $d_{k}=d_{n+1-k}$.  One 
valid solution for the operators $S_3^{(k,k+1)}$ is the $d_{k+1}\times d_k$ matrix whose main diagonal consists of the 
singular values of  $S_3^{(k,k+1)}$ and the rest of the elements are zero.  
From the relation between $S_3$ and $S_2$ 
given in Lemma~\ref{lem1} a similar solution can be constructed for $S_2$.

Homogeneously gapped Hamiltonians form only a subclass of operators that belong to $\mathfrak{su}(2)$.
Indeed, consider the block diagonal operator $H=\bigoplus_{m} H_m$, with 
\begin{align}\nonumber
H_m&=\sum_{k=1}^{d_m} \sum_{i=1}^{d_{k_m}}(\lambda_m-k\,c)\ket{k_m,i}\bra{k_m,i}\\
&\equiv\sum_{k=1}^{d_m} (\lambda_m-k\,c)\ket{k_m}\bra{k_m}\otimes\one_{d_{k_m}},
\label{eq:blockdiagonalsu2}
\end{align}
where $d_m$ is the dimension of the subspace upon which $H_m$ acts, $\lambda_m$ is largest eigenvalue in the 
homogeneously gapped spectrum of $H_m$, and $\ket{k_m, i}$ the corresponding eigenvectors.  The operator $H$ is 
not homogeneously gapped as for any two blocks, $m, n, \, \lambda_m-\lambda_n$ can be arbitrary.  However, as 
each block $H_m$ is homogeneously gapped, one can use Lemma~\ref{lem1} and Theorem~\ref{thm2} above to 
construct Hermitian operators $S_{2,m},\, S_{3,m}$ such that $\{H_m, \, S_{2,m}, S_{3,m}\}$ are the generators of 
$\mathfrak{su}(2)$ acting on the appropriate $d_m$-dimensional subspace.  Consequently, the operators $\{H,\, S_2=
\bigoplus_m S_{2,m}, \, S_3=\bigoplus S_{3,m}\}$ are the generators of $\mathfrak{su}(2)$ on the entire Hilbert space.
Constructing pretty good states for such Hamiltonians is also possible so long as at least one block has dimension  
$d_m\propto N$. 

Homogeneously gapped Hamiltonians form an important subclass of operators as it includes, but is not limited to, 
almost all local Hamiltonians 
studied in parameter estimation to date, as well as nearest-neighbor Hamiltonians that appear in interacting one-dimensional systems, graph state Hamiltonians~\cite{RJ09}, 
as well as Hamiltonians used in topological quantum computing~\cite{Freedman:03}.   
In the next section we illustrate how Theorems~\ref{thm1}, and~\ref{thm2} can be used to 
construct the $\mathfrak{su}(2)$ Lie algebra for a 
few of the afore mentioned Hamiltonians, and we also determine the states that yield a Heisenberg 
scaling in precision for noiseless quantum metrology.

\section{Examples of $\mathfrak{su}(2)$ Hamiltonians}
\label{sec:examples}

In the previous section we showed how one can construct the requisite Lie algebra from a homogeneously gapped  
Hamiltonian.  In this section we illustrate 
how the construction of Sec.~\ref{sec:Construction} works for four such 
Hamiltonians, the single body Hamiltonian, $H=1/2\sum_{i=1}^N \sz^{(i)}$ (Sec.~\ref{sec:local-Hamiltonian}), the 1-d 
cluster state Hamiltonian $H=\sum_{i=1}^N\sz^{(i-1)}\sx^{(i)}\sz^{(i+1)}$ (Sec.~\ref{sec:cluster}),  the
nearest-neighbor Hamiltonian, $H=\sum_{i=1}^{N-1}\sz^{(i)}\sz^{(i+1)}$ 
(Sec.~\ref{sec:nearest-neighbour-Hamiltonian}) and the Hamiltonian, 
$H=\sum_{i=1}^{N-1}\sy^{(i)}\sy^{(i+1)}+\sx^{\otimes N}+\sz^{\otimes N}$ (Sec.~\ref{sec:non-local-Hamiltonian}).
\vspace{-3mm}
\subsection{Local Hamiltonian}
\label{sec:local-Hamiltonian}

One of the most frequently used Hamiltonians in quantum metrology is the local Hamiltonian~\cite{GLM04,GLM06, Huelga:97, Escher:11, Kolodynski:12,*Kolodynski:13} 
\begin{equation}
 H=\frac{1}{2}\sum_{i=1}^N\sz^{(i)},
 \label{eq:local-Hamiltonian}
\end{equation}
whose spectrum and multiplicities are given by
\begin{align}\nonumber
\sigma(H)&=\left\{\lambda_x=\frac{N}{2}-x;\;\; d_x=\binom{N}{x},\;\; x\in(0,\ldots, N)\right\}.
\end{align}
Such a Hamiltonian frequently appears in the estimation of local field~\cite{Fleischhauer:00, Budker:07, Aiello:13}.

For ease of exposition we illustrate our construction for  $N=5$.  Using Theorem~\ref{thm2} the matrices $S_3^{(k,k+1)}$ are given by
\begin{align}\nonumber
&S_3^{(1,2)}=\begin{pmatrix}\sqrt{\frac{5}{4}}&0&0&0&0\end{pmatrix}\\ \nonumber
&S_3^{(2,3)}=\begin{pmatrix}\sqrt{\frac{8}{4}}&0&0&0&0&0&0&0&0&0\\
			    0&\sqrt{\frac{3}{4}}&0&0&0&0&0&0&0&0\\
			    0&0&\sqrt{\frac{3}{4}}&0&0&0&0&0&0&0\\
			    0&0&0&\sqrt{\frac{3}{4}}&0&0&0&0&0&0\\
			    0&0&0&0&\sqrt{\frac{3}{4}}&0&0&0&0&0
             \end{pmatrix}\\ \nonumber
&S_3^{(3,4)}=\begin{pmatrix}\frac{3}{2}&0&0&0&0&0&0&0&0&0\\
			    0&1&0&0&0&0&0&0&0&0\\
			    0&0&1&0&0&0&0&0&0&0\\
			    0&0&0&1&0&0&0&0&0&0\\
			    0&0&0&0&1&0&0&0&0&0\\
			    0&0&0&0&0&\frac{1}{2}&0&0&0&0\\
			    0&0&0&0&0&0&\frac{1}{2}&0&0&0\\
			    0&0&0&0&0&0&0&\frac{1}{2}&0&0\\
			    0&0&0&0&0&0&0&0&\frac{1}{2}&0\\
			    0&0&0&0&0&0&0&0&0&\frac{1}{2}
             \end{pmatrix}\\ \nonumber
&S_3^{(4,5)}=S_3^{(2,3)\dagger}\\
&S_3^{(5,6)}=S_3^{(1,2)\dagger}.
\label{eq:local-algebra}
\end{align}
Recalling that $S_2^{(k,k+1)}=-i S_3^{(k,k+1)}$, generator $S_2$ can easily be determined from 
Eq.~\eqref{eq:local-algebra}. Defining the ladder operators for $S_2$ as in Eq.~\eqref{eq:ladders} the eigenstate 
corresponding to the minimum eigenvalue of $S_2$ is, 
\begin{align}\nonumber
\ket{\psi_{\min}}&=\frac{-i}{4\sqrt{2}}\ket{\psi_{00000}}+\sqrt{\frac{5}{32}}\ket{\psi_{00001}}+\frac{i\sqrt{5}}{4}\ket{\psi_{00111}}\\
&-\frac{\sqrt{5}}{4}\ket{\psi_{10001}}-i\sqrt{\frac{5}{32}}\ket{\psi_{11011}}+\frac{1}{4\sqrt{2}}\ket{\psi_{11111}},
\label{eq:s2_min_eigenstate}
\end{align}
where $\{\ket{\psi_{\j}}\}$ denotes the basis in which $H$ is a diagonal matrix with its eigenvalues ordered from highest 
to lowest. For eigenvalues that are degenerate, we choose without loss of generality one eigenstate from the corresponding eigenspace. Raising the state in Eq.~\eqref{eq:s2_min_eigenstate} twice using $J^{(2)}_+$, yields the pretty good state
\begin{align}\nonumber
\ket{\psi_{PG}}&=\frac{i\sqrt{5}}{4}\ket{\psi_{00000}}-\frac{1}{4}\ket{\psi_{00001}}+\frac{i}{2\sqrt{2}}\ket{\psi_{00111}}\\
&-\frac{1}{2\sqrt{2}}\ket{\psi_{10001}}+\frac{i}{4}\ket{\psi_{11011}}-\frac{\sqrt{5}}{4}\ket{\psi_{11111}},
\label{eq:pg_state_local}
\end{align}
whose variance is $17/4$. The maximum possible variance is $25/4$ and is achieved by the GHZ state 
$\frac{1}{\sqrt{2}}\left(\ket{\psi_{00000}}+\ket{\psi_{11111}}\right)$, whereas the product state $\ket{+++++}$ achieves a variance of 
$\frac{5}{4}$.

The operators $S_2,\,S_3$ derived from our construction do not look like the standard $S_x,\,S_y$ operators.  
However, as mentioned above, the operators $S_2,\,S_3$ are one possible solution. One obtains $S_x$ ($S_y$) from $S_2$ ($S_3$) by conjugating the latter with the unitary that maps the eigenbasis of $S_x$ ($S_y$) to that of $S_2$ ($S_3$).  Applying the same unitary to the state in Eq.~\eqref{eq:pg_state_local} one obtains the
Dicke state $\ket{5, 2}$ which achieves a variance of $17/4$.  In the next sub-section we consider Hamiltonians based on graph states.

\subsection{One-dimensional cluster state Hamiltonian}
\label{sec:cluster}
Consider the one-dimensional cluster state Hamiltonian, $H=\sum_{i=1}^N\sz^{(i-1)}\sx^{(i)}\sz^{(i+1)}$, investigated by 
Rosenkratz and Jaksch~\cite{RJ09}.  This Hamiltonian can be easily obtained from the local Hamiltonian, $H=\sum_{i=1}^N\sx^{(i)}$, as
\begin{equation}
\sum_{i=1}^N\sz^{(i-1)}\sx^{(i)}\sz^{(i+1)}=V\left(\sum_{i=1}^{N}\sx^{(i)}\right)V^{\dagger},
\end{equation}
where 
\begin{equation}
 V=\prod_{i=1}^N U^{(i,i+1)}_\mathrm{ph}
\end{equation}
with 
\begin{equation}
U^{(i,i+1)}_\mathrm{ph}=\ketbra{0}{0}\otimes I^{(i+1)}+\ketbra{1}{1}\otimes\sz^{(i+1)}.  
\label{phasegate}
\end{equation}
In general, any graph state Hamiltonian, $H_G$ where $G=(V,E)$ is a graph whose vertices, $V$, correspond to 
physical qubits and edges $E$ between two vertices correspond to interactions,  can be written as 
$H_G=\sum_i K^{(i)}$, where $K^{(i)}$ are stabilizers~\cite{Hein:06}. All such Hamiltonians can be obtained from the local 
Hamiltonian $H=\sum_i \sx^{(i)}$ by conjugation with 
\begin{equation}
V=\prod_{i,j\in E} U^{(i,j)}_{\mathrm{ph}}, 
\label{eq:conj}
\end{equation}
where $U^{(i,j)}_{\mathrm{ph}}$ is the two qubit phase gate in Eq.~\eqref{phasegate} between any two qubits $i, j$ 
connected by an edge.  Hence, the construction of $S_2, \, S_3$  for the linear cluster state Hamiltonian proceeds by 
first constructing the corresponding operators for the local Hamiltonian, $H=\sum_{i=1}^N\sigma_x^{(i)}$, followed by 
conjugation by $V$.  Similarly the pretty good states for the one-dimensional cluster 
state Hamiltonian are obtained by applying $V$ on the pretty good state constructed for $H=\sum_{i=1}^N
\sigma_x^{(i)}$. 

\subsection{Nearest-neighbor Hamiltonian}
\label{sec:nearest-neighbour-Hamiltonian}
In this subsection we show how our construction works for the case where the dynamics of our quantum system is given 
by $U_{\theta}=\exp(i\theta H)$ with $H$ the \defn{nearest-neighbour} Hamiltonian 
\begin{equation}
H_{nn}=\sum_{i=1}^{N-1}\sz^{(i)}\sz^{(i+1)}.
\label{eq:nearest_neighbourH}
\end{equation}
It can be shown (see Appendix~\ref{append4}) that the spectrum of $H_{nn}$ in Eq.~\eqref{eq:nearest_neighbourH} is 
given by
\begin{align}\nonumber
\sigma(H_{nn})=&\left\{\lambda_x=N-1-2x; \;\; d_x=2 \binom{N-1}{x},\right.\\
&\left. x\in (0\ldots N-1)\right\},
\label{eq:non-local-Hamiltonian-spectrum}
\end{align}
Notice that, as the multiplicities of the Hamiltonians 
in Eqs.~(\ref{eq:local-Hamiltonian},~\ref{eq:nearest_neighbourH}) are not equal, there exist no real numbers 
$\alpha, \,\beta$ such that $\alpha H_{nn}+\beta\one=H$.   
Nearest-Neighbour Hamiltonians appear frequently in one-dimensional systems, and in metrology the goal is to estimate the interaction strength between neighboring qubits.

In order to explicitly illustrate the method of Section~\ref{sec:Construction}, let us consider the case $N=5$. Using 
Theorem~\ref{thm2} the matrices $S_3^{(k,k+1)}$ are given by
\begin{align}\nonumber
&S_3^{(1,2)}=\begin{pmatrix}2&0&0&0&0&0&0&0\\0&2&0&0&0&0&0&0\end{pmatrix}\\ \nonumber
&S_3^{(2,3)}=\left(\begin{array}{cccccccccccc}\sqrt{6}&0&0&0&0&0&0&0&0&0&0&0\\
			0&\sqrt{6}&0&0&0&0&0&0&0&0&0&0\\
			0&0&\sqrt{2}&0&0&0&0&0&0&0&0&0\\
			0&0&0&\sqrt{2}&0&0&0&0&0&0&0&0\\
			0&0&0&0&\sqrt{2}&0&0&0&0&0&0&0\\
			0&0&0&0&0&\sqrt{2}&0&0&0&0&0&0\\
			0&0&0&0&0&0&\sqrt{2}&0&0&0&0&0\\
			0&0&0&0&0&0&0&\sqrt{2}&0&0&0&0
	\end{array}\right)\\ \nonumber
&S_3^{(3,4)}=S_3^{(2,3)\dagger}\\
&S_3^{(4,5)}=S_3^{(1,2)\dagger},
\label{eq:constr_operators_non_local}  
\end{align}
and the corresponding matrices $S_2^{(k+1,k)}$ are obtained via the relation $S_2^{(k,k+1)}=-i S_3^{(k,k+1)}$.  

Having obtained the operators $S_2,\, S_3$ we can now apply Theorem~\ref{thm1} 
and determine states that yield Heisenberg-like scaling for noiseless parameter estimation. 
Choosing $H_{nn}=S_1$, and the ladder operators
\begin{equation}
 J^{(3)}_\pm=\frac{1}{\sqrt{2}}\left(S_1\pm iS_2\right)
 \label{eq:nearest_neighbour_ladders}
\end{equation}
the eigenstates of $S_3$ with the minimum eigenvalue are a $2$-fold degenerate subspace spanned by 
\begin{align}\nonumber
\ket{\Psi_1}&=\frac{1}{4}\ket{\psi_{00000}}-\frac{1}{2}\ket{\psi_{00010}}+\sqrt{\frac{3}{8}}\ket{\psi_{01010}}\\ \nonumber
&-\frac{1}{2}\ket{\psi_{10110}}+\frac{1}{4}\ket{\psi_{11110}}\\ \nonumber
\ket{\Psi_2}&=\frac{1}{4}\ket{\psi_{00001}}-\frac{1}{2}\ket{\psi_{00011}}+\sqrt{\frac{3}{8}}\ket{\psi_{01011}}\\
&-\frac{1}{2}\ket{\psi_{10111}}+\frac{1}{4}\ket{\psi_{11111}},
\label{eq:nearest-neighbour_s2-min-eigenstates}
\end{align}
where $\{\ket{\psi_{\j}}\}$ denotes the basis in which $H_{nn}$ is a diagonal matrix with its eigenvalues ordered from 
highest to lowest. As any state in the span of the states 
given in Eq.~\eqref{eq:nearest-neighbour_s2-min-eigenstates} can be used as our initial state, we choose without loss 
of generality $\ket{\Psi_1}$.  Applying the raising operator, $J^{(3)}_+$, twice to $\ket{\Psi_1}$ yields the normalized 
state 
\begin{equation}
\ket{\Phi_1}=\sqrt{\frac{3}{8}}\ket{\psi_{00000}}-\frac{1}{2}\ket{\psi_{01011}}+\sqrt{\frac{3}{8}}\ket{\psi_{11110}}.
\label{eq:nearest-neighbour_pg_state}
\end{equation}
The variance of the Hamiltonian with respect to this state is $12$.  The 
maximum possible variance of $H_{nn}$, achievable with the state 
$\ket{\psi_\mathrm{opt}}=\sqrt{\frac{1}{2}}(\ket{00000}+\ket{10101})$, 
is $16$.  Finally, the product state $\ket{+++++}$ achieves a variance of $4$.

\subsection{Non-local Hamiltonian}
\label{sec:non-local-Hamiltonian}

In order to illustrate our method for a Hamiltonian different from the single-body, graph state, and nearest-neighbour 
Hamiltonians, i.e.~that cannot be obtain by rescaling and shifting the spectrum of either the single-body 
Hamiltonian or the nearest-neighbour Hamiltonian, consider the Hamiltonian  
\begin{equation}
H_{nl}=\sum_{i=1}^N\sy^{(i)}\sy^{(i+1)}+\sx^{\otimes N}+\sz^{\otimes N}.
\label{eq:non-localH}
\end{equation}
Let us consider the case $N=4$, for which the eigenvalues and corresponding multiplicities are 
\begin{align}\nonumber
\lambda_x=& \{5,\,3,\,1,\,-1,\,-3,\,-5\},\\
d_x=&\{1,\,1,\,6,\,6,\,1,\,1\}.
 \label{nlspectrum}
\end{align}
We remark that the Hamiltonian in Eq.~\eqref{eq:non-localH} has a homogeneously gapped spectrum only if $N$ is 
even.

Applying Theorem~\ref{thm2} yields the following matrices for $S_3^{(k,k+1)}$
\begin{align}\nonumber
S_3^{(1,2)}&=\sqrt{5}\\ \nonumber
S_3^{(2,3)}&=\left(\begin{array}{ccccccc}
               \sqrt{8}&0&0&0&0&0
              \end{array}\right)\\ \nonumber
S_3^{(3,4)}&=\left(\begin{array}{ccccccc}
               3&0&0&0&0&0\\
               0&1&0&0&0&0\\
               0&0&1&0&0&0\\
               0&0&0&1&0&0\\
               0&0&0&0&1&0\\
               0&0&0&0&0&1
              \end{array}\right)\\ \nonumber
S_3^{(4,5)}&=S_3^{(2,3)\dagger}\\
S_3^{(5,6)}&=S_3^{(1,2)\dagger},
\label{eq:constr_operators_non-local}
\end{align}
and the corresponding matrices $S_2^{(k+1,k)}$ are again obtained via the relation $S_2^{(k,k+1)}=-i S_3^{(k,k+1)}$.

Choosing $H_{nl}=S_1$, and ladder operators $J^{(3)}_\pm$ as in Eq.~\eqref{eq:nearest_neighbour_ladders}
the eigenstate corresponding to the lowest eigenvalue for $S_3$ are
\begin{align}\nonumber
 \ket{\psi_\mathrm{min}}=&-\frac{1}{4\sqrt{2}}\ket{\psi_{0000}}+\sqrt{\frac{5}{32}}\ket{\psi_{0001}}-\frac{\sqrt{5}}{4}\ket{\psi_{0010}}\\
 &+\frac{\sqrt{5}}{4}\ket{\psi_{1000}}-\sqrt{\frac{5}{32}}\ket{\psi_{1110}}+\frac{1}{4\sqrt{2}}\ket{\psi_{1111}},
 \label{eq:s2_min_eigenstate_non-local}
\end{align}
where $\{\ket{\psi_{\j}}\}$ denotes the basis in which $H_{nl}$ is a diagonal matrix with its eigenvalues ordered from 
highest to lowest.  Applying the raising operator, $J^{(3)}_+$, twice on the state in 
Eq.~\eqref{eq:s2_min_eigenstate_non-local} yields the 
normalized state 
\begin{align}\nonumber
 J^{(3)2}_+\ket{\Psi}=&-\frac{\sqrt{5}}{4}\ket{\psi_{0000}}+\frac{1}{4}\ket{\psi_{0001}}+\frac{1}{2\sqrt{2}}\ket{\psi_{0010}}\\
 &-\frac{1}{2\sqrt{2}}\ket{\psi_{1000}}-\frac{1}{4}\ket{\psi_{1110}}+\frac{\sqrt{5}}{4}\ket{\psi_{1111}},
 \label{eq:pg_non-local}
\end{align}
whose variance with respect to $H_{nl}$ is $17$.  Note that the optimal variance for $H_{nl}$ is 25 and is achieved by 
the equal superposition of the maximum and minimum eigenstates of $H_{nl}$.  Finally, the product state $\ket{++++}$ 
achieves a variance of $4$.

\section{Pretty good states in the presence of local dephasing noise}
\label{sec:noise}

In this section we analyze the performance of pretty good states for noiseless metrology in the presence of local 
dephasing noise.  Specifically,  in Sec.~\ref{sec:analytic} we use the upper bound to the QFI of~\cite{Escher:11}, as 
discussed in Sec.~\ref{sec:bounds},
to analytically bound the performance of our pretty good states.  We find that for local 
dephasing noise and a local Hamiltonian, the bound scales at the SQL and, when considering a particular local Kraus 
decomposition of the CPTP map describing the dephasing noise, the bound of~\cite{Escher:11} 
yields the same result for a large variety of states~\cite{Froewis:14}.  Moreover, for local dephasing and a nearest-
neighbor Hamiltonian we show that the bound of~\cite{Escher:11} for the choice of local Kraus decompositions 
coincides with the QFI in the absence of noise, which is always an upper bound to the QFI.  
As a result, we determine the 
usefulness of pretty good states for noisy metrology by numerically 
evaluating their QFI.  As computation of the latter becomes intractable with increasing number of probe systems we 
compute the optimal QFI for $N\leq 12$ in the case of local noise and a nearest-neighbour Hamiltonian in 
Sec.~\ref{sec:numeric}.  The case of local noise and local Hamiltonian has already been reported by us 
elsewhere~\cite{Froewis:14}.

\subsection{Analytical bounds for metrology in the presence of local dephasing noise}
\label{sec:analytic} 

In this sub-section we provide a general expression for the bound derived in~\cite{Escher:11} 
(see also Sec.~\ref{sec:analytic}) for the case of phase 
estimation in the presence of local dephasing noise.   We then compute this bound with respect to the pretty good 
states constructed in Sec.~\ref{sec:PGS}.

Consider the phase estimation scenario where $N$ probes are subject to a unitary evolution $U(\theta)=e^{i\theta H}$, 
with $H\in\cB(\cH^{\otimes N})$ the total Hamiltonian acting on the $N$ probes. In addition, 
the $N$ probes are subject to local dephasing noise described by a CPT map $\cE$, such that 
$\cE\left[U(\theta)(\cdot)U^\dagger(\theta)\right]=U(\theta)\cE(\cdot)U^\dagger(\theta)$, 
i.e.~the unitary evolution commutes with the 
noise.  

Without loss of generality let us assume that the local dephasing acts along the $z$-axis of the Bloch sphere. Thus, for 
a single, two-level system local dephasing is described by a CPT map with Kraus operators
\begin{align}\nonumber
S_0=&\sqrt{p}\,\one\\
S_1=&\sqrt{1-p}\, \sz,
\label{eq:Kraus1}
\end{align}
where $p=\frac{1-e^{-\gamma t_0}}{2}$ with $\gamma$ denoting the strength of the noise and $t_0$ the time interval 
under which the system is subject to noise~\footnote{One can equivalently describe the evolution of the quantum 
system by the master equation $\partial_t\rho=i\theta[H,\rho]+\frac{\gamma}{2}(\sz\rho\sz-\rho)$, where the evolution 
occurs over the fixed time interval $t_0$.}. Consequently, the CPT map describing local dephasing of $N$ qubits is described by the Kraus operators 
\begin{equation}
S_{\bm m}=\bigotimes_{i=1}^NS_{m_i}
\label{eq:Krausn}
\end{equation}
where $\bm m\equiv m_1\ldots m_N$ and $m_i\in(0,1),\,\forall i\in(0,\ldots, N)$.  As the unitary dynamics commutes with local dephasing noise, it follows that the entire dynamical evolution can be described by a CPT map, with Kraus operators given by
\begin{equation}
\tilde{S}_{\bm m}(\theta)=U(\theta)S_{\bm m}.
 \label{eq:KrausN}
\end{equation}

As mentioned in Sec.~\ref{sec:ParEst}, an upper bound on the QFI is given by Eq.~\eqref{eq:bound}. Substituting the 
Kraus operators given by Eq.~\eqref{eq:KrausN} into Eq.~\eqref{eq:bound} yields the trivial upper bound bound 
$\mathcal{F}\leq 4(\Delta H)^2$, which is the QFI one obtains in the absence of uncorrelated dephasing noise.  
However, as 
any two Kraus decompositions for the same CPT map are unitarily related, one can write any Kraus decomposition of 
$\cE_\theta$ as
\begin{equation}
\Pi_{\bm n}(\theta)=\sum_{\bm m}V(\theta)_{\bm{nm}}\tilde{S}_{\bm m}(\theta).
\label{eq:KrausUN2}
\end{equation}
In order to minimize Eq.~\eqref{eq:bound} one must optimize over all unitary operators $V(\theta)$.  Such optimization 
can be performed using semi-definite programming~\cite{Kolodynski:12,*Kolodynski:13}.  For the case of quantum 
metrology in the presence of local dephasing noise, and with the Hamiltonian given by $H=1/2S_z$ Escher {\it et al.}
show that for the case where $N\to\infty$ it suffices to optimize over all Kraus operators that are unitarily related by
$V(\theta)=e^{i\alpha\theta B}$ where $B$ is a Hermitian operator and $\alpha$ is a free parameter that needs to be 
optimized. 

With $V(\theta)$ given as above, one can compute the operators $A_1$, $A_2$ of 
Eq.~\eqref{eq:escher-operators}.  Differentiating the Kraus operators in Eq.~\eqref{eq:KrausUN2} with respect to $\theta$ one obtains
\begin{align}\nonumber
&\frac{\mathrm{d}\Pi_{\bm n}(\theta)}{\mathrm{d}\theta}=\sum_{\bm m}\left[\frac{\mathrm{d}V(\theta)}
{\mathrm{d}\theta}\right]_{\bm{nm}}\tilde{S}_{\bm m}(\theta)+V(\theta)_{\bm{nm}}\frac{\mathrm{d}\tilde{S}_{\bm 
m}(\theta)}{\mathrm{d}\theta}\\
&=i\alpha\sum_{\bm m}\left[B V(\theta)\right]_{\bm{nm}}\tilde{S}_{\bm m}(\theta)+iH\sum_{\bm m}
V(\theta)_{\bm{nm}}\tilde{S}_{\bm m}(\theta),
\label{eq:KrausUN2diff}
\end{align}
where we have made use of the fact that (recall Eq.~\eqref{eq:KrausN})
\begin{align}\nonumber
\frac{\mathrm{d}V(\theta)}{\mathrm{d}\theta}&=i\alpha B\,V(\theta)\\
\frac{\mathrm{d}\tilde{S}_{\bm m}(\theta)}{\mathrm{d}\theta}&=iH\tilde{S}_{\bm m}(\theta). 
\end{align}
Using a similar calculation for $\mathrm{d}\Pi^\dagger_{\bm n}(\theta)/\mathrm{d}\theta$ 
one obtains for the operators $A_1,\,A_2$ of Eq.~\eqref{eq:escher-operators}
\begin{align}\nonumber
A_1&=H^2+2\alpha H\sum_{\bm{mn}}S_{\bm m}^\dagger B_{\bm{mn}}S_{\bm n}\\ \nonumber
&+\alpha^2\sum_{\bm{mn}}S_{\bm m}^\dagger \left[B^2\right]_{\bm{mn}}S_{\bm n}\\
A_2&=H+\alpha\sum_{\bm{mn}}S_{\bm m}^\dagger B_{\bm{mn}}S_{\bm n},
\label{eq:local-dephasing-escher-operators}
\end{align}
and the bound of Eq.~\eqref{eq:bound} is given by 
\begin{equation}
C_Q\left(\ketbra{\psi_S}{\psi_S},\,\Pi_{\bm m}(\theta)\right)=4\left((\Delta H)^2+2\alpha \Xi+\alpha^2\Omega\right),
\label{eq:local-dephasing-escher-bound}
\end{equation}
where 
\begin{align}\nonumber
\Xi&=\left\langle H \sum_{\bm{mn}}S_{\bm m}^\dagger B_{\bm{mn}}S_{\bm n}\right\rangle-\left\langle H\right\rangle \times\\ 
\nonumber
&\left\langle \sum_{\bm{mn}}S_{\bm m}^\dagger B_{\bm{mn}}S_{\bm n}\right\rangle\\
\Omega&=\left\langle \sum_{\bm{mn}}S_{\bm m}^\dagger \left[B^2\right]_{\bm{mn}}S_{\bm n}\right\rangle-\left\langle 
\sum_{\bm{mn}}S_{\bm m}^\dagger B_{\bm{mn}}S_{\bm n}\right\rangle^2.
\label{eq:local-dephasing-escher-special-operators}
\end{align}
The minimum of Eq.~\eqref{eq:local-dephasing-escher-bound} over the Kraus operators 
$\{\Pi_m(\theta)\}$ occurs for $\alpha_{\min}=-\Xi/\Omega$, and is given 
by 
\begin{equation}
C_Q^{\mathrm{min}}\left(\ketbra{\psi_S}{\psi_S},\,\Pi_{\bm m}(\theta)\right)=4\left((\Delta H)^2-\frac{\Xi^2}{\Omega}\right).
\label{eq:local-dephasing-escher-lower-bound}
\end{equation}

Our goal is to use the bound given in Eq.~\eqref{eq:local-dephasing-escher-lower-bound} to gauge the performance of pretty 
good states for noiseless metrology in the presence of local dephasing noise. We first focus on the case of pretty good 
states for the local Hamiltonian of Eq.~\eqref{eq:local-Hamiltonian}.  As the noise acts locally on each of the $N$ 
qubits, Escher {\it et al.} restrict their search for the Kraus decomposition that minimizes Eq.~\eqref{eq:bound} to local Kraus decompositions.  To that end they assume
$V(\theta)=e^{i\alpha\theta S_x}$.  For this choice of $V(\theta)$ it can be shown that 
(see Appendix~\ref{append5})
\begin{equation}
C_Q^{\mathrm{min}}=\frac{4N(\Delta H)^2(1-q^2)}{N(1-q^2)\one+q^2(\Delta H)^2},
\label{eq:local-dephasing-localH-bound}
\end{equation}
where $q\equiv 4p(1-p)$.

Indeed Escher {\it et al.} go on to show that in the limit of large $N$ the upper bound in 
Eq.~\eqref{eq:local-dephasing-localH-bound} coincides with the well-known lower bound on the QFI 
found in~\cite{OK01}.  However, there is no 
{\it a priori} reason to believe that for finite $N$, the bound in Eq.~\eqref{eq:local-dephasing-localH-bound} 
is minimized by local Kraus operators.    As the CPT map describing local dephasing noise is 
permutation invariant, perhaps a better bound for the case of finite $N$ can be obtained for Kraus 
operators  that respect the permutation symmetry of the map.  One 
possible choice of such Kraus operators are 
\begin{equation}
\tilde{K}_{\bm n}(\theta)=U(\theta)^{\otimes N}\sum_{\bm m} C_{\bm{nm}} S_{\bm m},
\label{Eq:sntensors}
\end{equation}
where $C$ is the Clebsch-Gordan transform~\cite{BCH06}.  The Kraus operators defined in Eq.~\eqref{Eq:sntensors} 
are collective operators acting on all $N$ qubits, and arise if we consider all $N$ qubits as coupling to a single $2^N$-
dimensional environment which we then trace over (see Eq.~\eqref{eq:purification}). However, as the Clebsch-Gordan 
transform is independent of $\theta$, the upper bound achievable by the Kraus operators in Eq.~\eqref{Eq:sntensors} 
is the same as the one achieved by the local Kraus decomposition.

We now determine the usefulness of our pretty good states from Sec.~\ref{sec:PGS} by computing the bound of 
Eq.~\eqref{eq:local-dephasing-localH-bound}.  The variance, $(\Delta H)^2_{\mathrm{PG}}$, of $H$ in 
Eq.~\eqref{eq:local-Hamiltonian} with respect 
to our pretty good states of $N$ qubits can be easily calculated to be $(\Delta H)^2_{\mathrm{PG}}=\frac{N}{4}(\frac{N}{2}+1)$, 
and the bound of 
Eq.~\eqref{eq:local-dephasing-localH-bound} is given by 
\begin{equation}
C_Q^{\mathrm{min}}(\mathrm{PG})=\frac{4N\left(\frac{N}{2}+1\right)\left(1-q^2\right)}{4(1-q^2)+q^2\left(\frac{N}{2}+1\right)}.
\label{eq:boundlocalpgs}
\end{equation}
In the limit $N\to\infty$, $C_Q(\mathrm{PG})\to \frac{4N(1-q^2)}{q^2}$.  If, however, 
we compute the bound of Eq.~\eqref{eq:local-dephasing-localH-bound} with respect to the GHZ state, one easily finds
\begin{equation}
C_Q(\mathrm{GHZ})=\frac{4N^2\left(1-q^2\right)}{4(1-q^2)+q^2N},
\label{eq:boundlocalghz}
\end{equation}
which, in the limit $N\to\infty$ also tends to $\frac{4N(1-q^2)}{q^2}$.

In fact, one can infer from Eq.~\eqref{eq:local-dephasing-localH-bound} that in the limit of large $N$, $C_Q$ 
will always scale as the SQL, up to some pre-factor.  For the optimal product state ($(\Delta H)^2= N/4$), this pre-
factor is $1-q^2$ whereas for our pretty good states the pre-factor is $(1-q^2)/q^2$.  This is the same pre-factor 
as for the GHZ state, which is known to yield the same precision as product states for 
large enough $N$ in the case of frequency estimation ($q=e^{-2\gamma t}$).  Hence, whereas the bound 
of Eq.~\eqref{eq:local-dephasing-localH-bound} gives the right scaling in 
precision for the case of phase estimation using local Hamiltonians in the presence of local dephasing noise, it is 
incapable of discriminating which states attain this bound.   

We now consider the case of phase estimation using the nearest-neighbour Hamiltonian 
(Eq.~\eqref{eq:nearest_neighbourH}) and local dephasing noise. We choose the same Kraus decomposition as 
Escher {\it et al.}, which lead to a good bound of the QFI in the case of local Hamiltonians. Our reason for choosing this particular Kraus decomposition is that the noisy process we consider  shares the same characteristics in both cases, namely local noise whose generators commute with the unitary dynamics.  In contrast, we will show 
that the bound obtained for the nearest-neighbour Hamiltonian considered here coincides with the QFI in the absence 
of any noise which is a trivial upper bound to the QFI. 

Using the local Kraus decomposition of the CPT map, $\tilde{S}_{\bm m}(\theta)$ gives  
\begin{align}\nonumber
&C_Q(\ket{\psi},\tilde{S}_{\bm m}(\theta))=\\
&4\left((\Delta H_{nn})^2-\frac{q^2(\langle H_{nn}S_z\rangle-\langle H_{nn}\rangle\langle S_z\rangle)^2}{N(1-q^2)\one
+q^2(\Delta H_{nn})^2}\right),
\label{eq:Escherboundnearestneighbour}
\end{align}
The states $\{\ket{\bm m}, \, \sx^{\otimes N}\ket{\bm m}\}$ are both eigenstates of $H_{nn}$ corresponding to the same eigenvalue, where $S_z\ket{\bm m}=m\ket{\bm m}$, 
$S_z\sx^{\otimes N}\ket{\bm m}=-m\sx^{\otimes N}\ket{\bm m}$.  
As $\alpha \ket{\bm m}+\beta\sx^{\otimes N}\ket{\bm m}$ is an eigenstate of $H_{nn}$ for all $\alpha,\beta$ satisfying $|\alpha|^2+|\beta|^2=1$, it follows that $(\Delta H_{nn})^2$ is invariant 
for any choice of $\alpha,\, \beta$, and that for $|\alpha|=|\beta|$,  $\langle H_{nn}S_z\rangle=\langle H_{nn}\rangle\langle S_z\rangle=0$.  Thus,
the maximum value of Eq.~\eqref{eq:Escherboundnearestneighbour} is $C_Q=4(\Delta H_{nn})^2$ and is achieved by choosing the state $\frac{1}{\sqrt{2}}\left(\ket{\psi_{\mathrm{max}}}+e^{i\phi}\ket{\psi_{\mathrm{min}}}\right)$, where $\phi\in(0,2\pi]$ and $\ket{\psi_{\mathrm{max}(\mathrm{min})}}$ are eigenstates of $H_{nn}$ belonging to the doubly degenerate subspaces corresponding to the maximum (minimum) eigenvalue, with $|\alpha|=|\beta|=1/\sqrt{2}$. 

Moreover, as both the minimum and 
maximum eigenspaces of the Hamiltonian in Eq.~\eqref{eq:nearest_neighbourH} are doubly degenerate, we can 
always choose $|\alpha|=|\beta|=1/\sqrt{2}$ for which Eq.~\eqref{eq:Escherboundnearestneighbour} gives the trivial 
bound $C_Q=4(\Delta H_{nn})^2$ for both the optimal and pretty good states.

The reason why the bound of Eq.~\eqref{eq:local-dephasing-escher-bound} is trivial for phase estimation using a 
nearest-neighbour Hamiltonian in the presence of local dephasing noise is due to the degeneracy of the spectrum of 
$H_{nn}$ in Eq.~\eqref{eq:nearest_neighbourH}.  States with different eigenvalue of $S_z$ belong to the same 
eigenspace of $H_{nn}$, and the numerator of the second term in Eq.~\eqref{eq:Escherboundnearestneighbour} can be optimized independently of $(\Delta H_{nn})^2$. Hence,  we are free to choose the states within a given eigenspace of $H_{nn}$ such that
the numerator in the second term of Eq.~\eqref{eq:Escherboundnearestneighbour} equals zero. One may argue that restricting the 
search for the optimal Kraus decomposition over local Kraus operators in this instance is a bad one.  One suitable 
choice could be to choose $V=e^{i\alpha\theta S_3}$ in Eq.~\eqref{eq:KrausUN2}, where $S_3$ is the Hermitian 
operator obtained by our construction in Sec.~\ref{sec:Construction}.

\subsection{QFI for pretty good states under local dephasing noise}
\label{sec:numeric} 

In this subsection we compare the QFI of pretty good states using local and nearest-neighbour Hamiltonians to that of 
the product and optimal states in the presence of local dephasing noise. We find that, 
for both scenarios and moderate number of probe systems, $N$, the pretty good states constructed in 
Secs.~\ref{sec:local-Hamiltonian} perform better than product states, but are far from the true optimal states.  The performance of pretty good states for local Hamiltonians under various types of local, as well as correlated noise, has been investigated elsewhere~\cite{Froewis:14}.

Using Eqs.~(\ref{eq:QFI1},~\ref{eq:SLD}) the QFI in the presence of local noise is
\begin{equation}
\mathcal{F}(\rho(\theta))=4\sum_{i<j} \frac{(\lambda_i-\lambda_j)^2}{\lambda_i+\lambda_j} \left| \left\langle \psi_i \right|  
H \left| \psi_j \right\rangle  \right|^2,
\label{fisherdephasing}
\end{equation}
where $\rho(\theta)\equiv\cE_\theta[\rho]=\sum_i\lambda_i\ketbra{\psi_i}{\psi_i}$.
Here, we study the performance of the pretty good states for frequency estimation, where the QFI obtained per unit 
time, ${\cal F}(\rho_\lambda(t))/t$, has to be optimized over time leading to an optimal interrogation time $t_{\rm opt}$.
We consider the relative improvement between entangled input states and the optimal product state $| + \rangle ^{\otimes N}$,
\begin{equation}
\label{eq:1}
I_{\mathrm{rel}}(\psi) = \frac{\max_t\mathcal{F}(\mathcal{E}_{\theta}(\left|\psi  \right\rangle\!\left\langle \psi\right| ))/t}{\max_t\mathcal{F}(\mathcal{E}_{\theta}(\left|+  \right\rangle\!\left\langle +\right|^{\otimes N} ))/t}.
\end{equation}
In addition, we calculate $I_{\mathrm{rel}}(\psi)$ for the optimal state in the absence and in the presence of noise. Whereas the former is simply the equally weighted superposition of the eigenstates with smallest and largest eigenvalue of $H$, the latter has to be numerically determined. To reduce the computational effort, we restrict ourselves to the totally symmetric subspace, i.e.~the subspace spanned by $J_+^{m} \left| \psi_{\mathrm{min}} \right\rangle , m \in \left\{ 0,\dots,2j_{\mathrm{max}} \right\}$, where $| \psi_{\mathrm{min}} \rangle $ is the ground state of $H$ and  $J_+$ is the ladder operator that creates excitations in the spectrum of $H$. In case of the nearest-neighbour Hamiltonian, $H_{nn}$, given by Eq.~\eqref{eq:nearest_neighbourH} we use the operators
\begin{align}\nonumber
S_2 &= i \sum_{j=1}^{N-1} \sigma_x^{\otimes j-1}\otimes \sigma_y\otimes\sigma_z \otimes \one^{\otimes N-j-1},\\
S_3 &= \sum_{j= 1}^{N-1} \sigma_x^{\otimes j} \otimes \one^{\otimes N-j}
\label{eq:alternative-operators}
\end{align}
to define $J_{+}$. One can easily check that the set $\{H_{nn},\,S_2,\, S_3\}$ is a valid choice of generators of $\mathfrak{su}(2)$. As $H_{nn}$ has a doubly degenerate ground energy spectrum, any state 
$\ket{\psi_{\min}(\alpha)}= \cos(\alpha/2) \ket{0101\dots} +  \sin(\alpha/2) \ket{1010\dots}$, for $\alpha \in \mathbbm{R}$, 
is a ground state of $H_{nn}$. Here, we fix $\alpha = \pi/2$, as this particular ground state is invariant under collective spin flips $\sigma_x^{\otimes N}$, which is a symmetry of $H_{nn}$.   This choice of $\alpha$ turns out  to numerically maximize $I_{\mathrm{rel}}$.

In the case of local Hamiltonians, $I_{\mathrm{rel}}$ was numerically computed for a variety of state families, including the pretty good states of Eq.~\eqref{eq:dickestate}, in~\cite{Froewis:14}. It was shown that these states outperform both the product and GHZ states, but perform significantly worse than the optimal states. Note that the best results are not achieved with $k= \lfloor j_{\mathrm{max}} \rfloor = \lfloor N/2 + 1 \rfloor$ (see Eq.~(\ref{eq:thm1})), but with a smaller $k$, which yields a reduced performance in the noiseless scenario.

For the nearest-neighbour Hamiltonian, we find a similar picture for small $N$ (see Fig.~\ref{fig:NumericalResults}). While pretty good states improve the metrological sensitivity, there is a gap to the performance of optimal states and, in this case, also to the optimal states in the absence of noise. Again, the optimal excitation $k$ is generally not $\lfloor j_{\mathrm{max}} \rfloor = \lfloor N/2 \rfloor$, but smaller (see Fig.~\ref{fig:3}).
 
\begin{figure}[htbp]
\centerline{\includegraphics[width=\columnwidth]{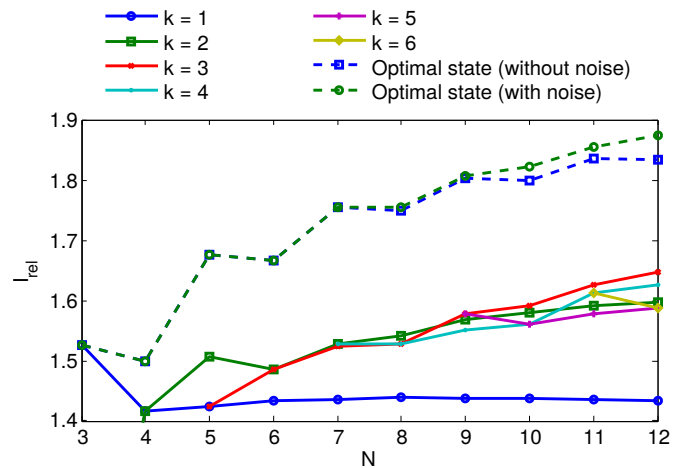}}
\caption[]{\label{fig:NumericalResults} Relative improvement $I_{\mathrm{max}}$ of different state families for the scenario nearest-neighbour Hamiltonian plus dephasing noise. The value of $k$ refers to the number of excitations above the ground state space of $S_3$ (see Eq.~(\ref{eq:thm1})). Note that $| + \rangle ^{\otimes N}$  is the ground state of $S_3$. For all $0<k\leq \lfloor N/2\rfloor$, we find increased performance compared to the product states. However, there exist states that give rise to higher sensitivity; for example, the optimal states in the absence of noise. In contrast to the scenario with local Hamiltonian and dephasing, these states are --at least for small $N$-- close to the actual optimal states found by numerical algorithms.}
\label{fig:3}
\end{figure}

\section{Conclusion}
\label{sec:conclusion}

In this work we use Lie algebraic techniques to construct states that achieve Heisenberg scaling in precision for a 
class of Hamiltonians, namely those with a homogeneously gapped spectrum, that satisfy the $\mathfrak{su}(2)$ Lie 
algebra.  This is a subclass of all Hamiltonians that satisfy the $\mathfrak{su}(2)$ Lie algebra and includes local, nearest-neighbor, graph state, and topological Hamiltonians many 
of which play an important role in quantum metrology. For Hamiltonians with a homogeneously gapped spectrum  
we identify necessary conditions regarding the multiplicities of $H$, for any Hermitian operator to be a valid generator 
of the $\mathfrak{su}(2)$ Lie algebra.

We also investigate the performance of the states constructed by our procedure in the presence of local dephasing 
noise.  Specifically, we calculate a well-known bound on the QFI~\cite{Escher:11} and find that, for the case of a local 
Hamiltonian, our states achieve the SQL.  However, for a particular local Kraus decomposition of the 
CPTP map describing local dephasing noise we discover that the bound of~\cite{Escher:11} 
yields the same bound for a variety of states.  Moreover, in the case of nearest-neighbor Hamiltonians the bound 
of~\cite{Escher:11}, restricted to a particular local Kraus decomposition, is equal to the QFI in the absence of noise, and thus does not 
yield an informative bound.  A tighter bound in this case can be obtained when optimizing over all Kraus operators and 
not just the local ones. 

We also numerically determined the actual QFI of the states constructed by our techniques for the case of nearest-neighbor Hamiltonians (for the results of the local Hamiltonian, see Ref.~\cite{Froewis:14}). We discover that---similar 
to the scenario with local Hamiltonians---pretty good states outperform product states. However, they are suboptimal as 
there exist other states providing higher sensitivity.  It would be interesting to analyze the states of our construction 
under the influence of other types of noise.

Several important questions arise with regards to establishing insightful upper bounds on the QFI.  For example, 
in searching for the optimal Kraus decomposition in the case where $H=S_z$ and local dephasing noise, 
Escher {\it et al.} optimized over the unitary generated by $S_x$.  If the Hamiltonian is a generator of 
$\mathfrak{su}(2)$ one may ask whether in searching for the optimal Kraus decomposition, it is sufficient to search 
over all unitary re-mixings of the Kraus operators generated by either of the remaining generators of the algebra.  

\section*{Acknowledgements}
This work was supported by the Austrian Science Fund (FWF): Grant numbers P24273-N16, Y535-N16, and J3462.
\appendix
\section{Proof of Theorem~\ref{thm1}}
\label{append2}

In this appendix we provide a proof of Theorem~\ref{thm1} regarding the construction of pretty good states for 
noiseless metrology. 
\begin{reptheorem}{thm1}
Let $\{S_1,\,S_2,\,S_3\}$ be a set of generators for $\mathfrak{su}(2)$. Assume, without loss of generality, 
that $H\equiv S_1$, and let $\ket{\psi_\mathrm{min}}$ be an eigenstate of $S_3$ corresponding to the smallest 
eigenvalue.  Then the variance of $H$ with respect to the state 
\begin{equation}
\ket{\psi}=\sqrt{\frac{1}{\mathcal{N}}} J_+^{(3)k}\ket{\psi_\mathrm{min}},
\end{equation} 
where $\mathcal{N}$ denotes the normalization constant, scales as half the 
\defn{spectral radius} of $\bm J^2$, $\varrho(\bm J^2)$, if $k=\lceil\frac{2j_\mathrm{max}+1}{2}\rceil$, where $j_
\mathrm{max}$ is related to the maximum eigenvalue of $\bm J^2$ via Eq.~\eqref{eq:angular-momentum}, and $\lceil\cdot\rceil$ is the ceiling function.
\end{reptheorem}

\begin{proof}
Recall that $(\Delta H)^2$, with respect to a state $\ket{\psi}$, is given by
\begin{equation}
(\Delta H)^2=\bra{\psi}H^2\ket{\psi}-\bra{\psi}H\ket{\psi}^2.
\label{eq:B1}
\end{equation}
Now let $\ket{\psi}=\sqrt{\frac{1}{\mathcal{N}}}\, J_+^{(3)k}\ket{\psi_\mathrm{min}}$, 
where $\mathcal{N}=\bra{\psi_\mathrm{min}}J_-^{(3)k}J_+^{(3)k}\ket{\psi_\mathrm{min}}$, with 
\begin{equation}
 S_3\ket{\psi_\mathrm{min}}= m_{\min}\ket{\psi_{\min}}
 \label{eq:B2}
\end{equation}
and recall that the eigenvalues, $m$, of $S_3$ lie within the range $-j\leq m\leq j$. As 
$\ket{\psi_{\min}}$ is an eigenstate of $S_3$ with the smallest possible eigenvalue, it follows that 
$m_{\min}=-c\,j_{\max}(j_{\max}+1)$, where $c\,j_{\max}(j_{\max}+1)$ is the maximum possible value of $\bm J^2$.  

Substituting $\ket{\psi}$ into Eq.~\eqref{eq:B1}, and noting that $J_+^{(3)\dagger}=J_-^{(3)}$, yields
\begin{align}\nonumber
(\Delta H)^2=&\frac{1}{\mathcal{N}}\bra{\psi_\mathrm{min}}J_-^{(3)k}H^2J_+^{(3)k}\ket{\psi_\mathrm{min}}\\
&-\frac{1}{\mathcal{N}^2}\bra{\psi_\mathrm{min}}J_-^{(3)k}HJ_+^{(3)k}\ket{\psi_\mathrm{min}}^2.
\label{eq:B3}
\end{align}
Using $H\equiv S_1=1/\sqrt{2}(J^{(3)}_++J^{(3)}_-)$ (see Eq.\eqref{eq:ladders}) we have
\begin{align}\nonumber
J_-^{(3)k}H^2J_+^{(3)k}&=\frac{1}{2}\left(J_-^{(3)k}J_+^{(3)k+2}+J_-^{(3)k}J_+^{(3)}J_-^{(3)}J_+^{(3)k}\right.\\ 
\nonumber
&\left.+J_-^{(3)k+1}J_+^{(3)k+1}+J_-^{(3)k+2}J_+^{(3)k}\right),\\
J_-^{(3)k}HJ_+^{(3)k}&=\frac{1}{\sqrt{2}}\left(J_-^{(3)k}J_+^{(3)k+1}+J_-^{(3)k+1}J_+^{(3)k}\right).
\label{eq:B4}
\end{align}
As $J_+^{(3)m}\ket{\psi_{\min}}$ is an eigenstate of $S_3$, with eigenvalue $c(-j_{\max}+m)$, it follows that 
\begin{equation}
\bra{\psi_{\min}}J_-^{(3)m}J_+^{(3)n}\ket{\psi_{\min}}\propto\delta_{mn}
\label{eq:B5}
\end{equation}
as either $J_+^{(3)n}\ket{\psi_{\min}}$, $J_+^{(3)m}\ket{\psi_{\min}}$ correspond to different 
eigenvalues of $S_3$ or
\begin{align*}
J_+^{(3)n}\ket{\psi_{\min}}=&0,\\
J_+^{(3)m}\ket{\psi_{\min}}=&0.
\end{align*}
Hence, Eq.~\eqref{eq:B3} reduces to 
\begin{align}\nonumber
(\Delta H)^2&=\frac{1}{2\mathcal{N}}\left(\bra{\psi_{\min}}J_-^{(3)k}\{J_-^{(3)},J_+^{(3)}\}J_+^{(3)k}\ket{\psi_{\min}}\right)\\ 
\nonumber
&=\frac{1}{2\mathcal{N}}\left(\bra{\psi_{\min}}J_-^{(3)k}\bm J^2J_+^{(3)k}\ket{\psi_{\min}}\right.\\
&\left.-\bra{\psi_{\min}}J_-^{(3)k}S_3^2J_+^{(3)k}\ket{\psi_{\min}}\right)
\label{eq:B6}
\end{align}
where we have made use of the expression $\bm J^2=S_3^2+\{J_+^{(3)},J_-^{(3)}\}$.  As
$\frac{1}{\mathcal{N}}\bra{\psi_{\min}}J_-^{(3)k}S_3^2J_+^{(3)k}\ket{\psi_{\min}}=\frac{c^2}{\mathcal{N}}(-j_{\max}+k)^2
$ , and $\bra{\psi_{\min}}J_-^{(3)k}\bm J^2J_+^{(3)k}\ket{\psi_{\min}}=c^2j_{\max}(j_{\max}+1)=\varrho(\bm J^2)$, we would like 
to choose $k$ such that the second term of Eq.~\eqref{eq:B6} is as small as possible (i.e.~$
\mathcal{O}(1)$).  This occurs for states, $J^{(3)k}_+\ket{\psi_\mathrm{min}}$, whose $S_3$ eigenvalue is close to zero.  Hence, we simply apply 
$J_+^{(3)}$ a number of times equal to 
\begin{equation}
k=\left\lceil\frac{2j_{\max}+1}{2}\right\rceil.
\label{eq:B7}
\end{equation}
Finally, note that $\bra{\psi_\mathrm{min}}J_-^{(3)k}J_+^{(3)k}\ket{\psi_\mathrm{min}}=\mathcal{N}$ which cancels the 
normalization in Eq.~\eqref{eq:B6}.  This completes the proof.
\end{proof}

\section{Proofs of Lemma~\ref{lem1} and Theorem~\ref{thm2}}
\label{append3}

In this appendix we determine a class of Hamiltonians for which the construction of pretty good states for noiseless 
metrology given in Sec.~\ref{sec:PGS} applies.  Given a Hamiltonian, $H\equiv S_1$, and assuming that $H$ has 
homogeneously gapped spectrum, we show that two Hermitian operators $S_2, \, S_3$, such that 
$\{S_1,S_2,S_3\}$ are generators of $\mathfrak{su}(2)$ must be of a particular form (Lemma~\ref{lem1}).  
We then show in Theorem~\ref{thm2} that in order to determine the Hermitian operators $S_2$, $S_3$ such 
that $\{S_1,S_2,S_3\}$ are generators of $\mathfrak{su}(2)$ the multiplicities, $d_k$, of the homogeneously gapped 
spectrum of eigenvalues, $\lambda_k$, of $S_1$ must necessarily obey the conditions $d_{k+1}\geq d_k$ and 
$d_k=d_{n+1-k}$ for all $k$.  In addition, Theorem~\ref{thm2} also provides one possible choice for the Hermitian operators $S_2$ and $S_3$.   

Throughout this appendix we will assume the operators $S_2, \, S_3$ are Hermitian.  Furthermore, we will assume 
without loss of generality that the eigenvalues of $S_1$ are arranged in  decreasing order, 
i.e. $\lambda_1>\lambda_2>\ldots>\lambda_N$ and that the spectrum of $S_1$ is homogeneously gapped. 

We begin by proving Lemma~\ref{lem1}
\begin{replemma}{lem1}
Let $S_1$ be given as in Eq.~\eqref{eq:spectraldecomp} and let $S_2,\, S_3$ be two Hermitian operators. 
If the spectrum of $S_1$ is homogeneously gapped, i.e.~$|\lambda_{k+1}-\lambda_{k}|=c,\,\forall k$, and the 
conditions in Eq.~\eqref{eq:Lie-algebra} hold then
\begin{align}\nonumber
S_2&=\sum_{k=1}^n\ket{k+1}\bra{k}\otimes S_2^{(k,k+1)}+\ket{k}\bra{k+1}\otimes S_2^{(k+1,k)}\\
S_3&=\sum_{k=1}^n\ket{k+1}\bra{k}\otimes S_3^{(k,k+1)}+\ket{k}\bra{k+1}\otimes S_3^{(k+1,k)},
\label{eq:C1}
\end{align}
where $S_2^{(k,l)}$ ($S_3^{(k,l)}$) are $d_l\times d_k$ matrices, and $S_2^{(k,k+1)}=-i S_3^{(k,k+1)},\,\forall k$.
\end{replemma}

\begin{proof}
Substituting the third equation of Eq.~\eqref{eq:Lie-algebra} into the second gives
\begin{equation}
c^2S_2=H^2S_2+S_2H^2-2HS_2H.
\label{eq:C2}
\end{equation}
Write 
\begin{equation}
S_2=\sum_{k,l=1}^n\ket{k}\bra{l}\otimes S_2^{(l,k)},
\label{eq:C3}
\end{equation}
where $S_2^{(l,k)}$ is a $d_k\times d_l$ matrix.   As $S_2$ is Hermitian by assumption, 
$S_2^{(l,k)\dagger}=S_2^{(k,l)}$.  Plugging Eqs.~(\ref{eq:spectraldecomp},~\ref{eq:C3}) into Eq.~\eqref{eq:C2} one obtains, after some algebra,
\begin{equation}
0=\sum_{k,l}\left[(\lambda_k-\lambda_l)^2-c^2\right]\ket{k}\bra{l}\otimes S_2^{(l,k)}. 
\label{eq:C5}
\end{equation}
As the spectrum of $S_1$ is homogeneously gapped by assumption it follows that 
\begin{enumerate}
\item $S_2^{(k,k)}=0$.
\item For $\lambda_{k-1}-\lambda_k=c$ and $\lambda_{k+1}-\lambda_k=-c$, 
$S_2^{(k,k-1)}$ and $S_2^{(k,k+1)}$ can be arbitrary.
\item For $m>1$, $S_2^{(k\pm m,k)}=0$.
\end{enumerate}
Hence, the only non-zero matrices in Eq.~\eqref{eq:C3} are those immediately above and below the main diagonal, 
i.e.
\begin{equation}
S_2=\sum_{k}^{n-1}\ket{k+1}\bra{k}\otimes S_2^{(k,k+1)}+\ket{k}\bra{k+1}\otimes S_2^{(k+1,k)}.
\label{eq:C6}
\end{equation}
Following similar arguments as above one finds that 
\begin{equation}
S_3=\sum_{k}^{n-1}\ket{k+1}\bra{k}\otimes S_3^{(k,k+1)}+\ket{k}\bra{k+1}\otimes S_3^{(k+1,k)}.
\label{eq:C7}
\end{equation}

We now show how the matrices $S_2^{(k,k+1)}$ and $S_3^{(k,k+1)}$ are related. 
Plugging Eqs.~(\ref{eq:C6},~\ref{eq:C7}) into the first equation of Eq.~\eqref{eq:Lie-algebra} one obtains, after some 
algebra,
\begin{align}\nonumber
&\sum_k\left[ (\lambda_{k+1}-\lambda_k)\ketbra{k+1}{k}\otimes S_2^{(k,k+1)}\right.\\ \nonumber
&\left.+(\lambda_{k}-\lambda_{k+1})\ketbra{k}{k+1}\otimes S_2^{(k+1,k)}\right]\\
&=i c\sum_k \ketbra{k+1}{k}\otimes S_3^{(k,k+1)}+\ketbra{k}{k+1}\otimes S_3^{(k+1,k)}.
\label{eq:C8}
\end{align} 
As $\lambda_{k}-\lambda_{k+1}=c,\, \forall k$ by assumption, it follows that $S_2^{(k,k+1)}=-iS_3^{(k,k+1)}$. 
This completes the proof. 
\end{proof}

We now prove Theorem~\ref{thm2}.
\begin{reptheorem}{thm2}
Let $S_1$ be given by Eq.~\eqref{eq:spectraldecomp} with the eigenvalues of $S_1$ satisfying 
$-\lambda_k=\lambda_{n-k+1}, \,\lambda_{k}-\lambda_{k+1}=c, \,\forall k\in(1,\ldots, n)$. 
In addition, let the operators $S_2,\, S_3$ be given as in Lemma~\ref{lem1}. Necessary conditions for 
Eq.~\eqref{eq:Lie-algebra} to hold are that $d_{k+1}\geq d_{k}$ for $1\leq k\leq\lfloor\frac{n}{2}\rfloor$, 
and  $d_{k}=d_{n+1-k}$.  Furthermore, one possible solution for the matrices $S_3^{(k,k+1)}$ is given by the 
$d_{k+1}\times d_k$ matrix 
\begin{align}\nonumber
S_3^{(k,k+1)}&=\sqrt{\frac{c}{2}}\,\mathrm{diag}\left(\underbrace{\sqrt{\sum_{i=1}^k\lambda_i}}_{d_1\,\mathrm{times}},
\underbrace{\sqrt{\sum_{i=2}^k\lambda_i}}_{(d_2-d_1)\,\mathrm{times}},\ldots,\right.\\
& \left.\underbrace{\sqrt{\sum_{i=k-1}^k\lambda_i}}_{(d_{k-1}-d_{k-2})\,\mathrm{times}},\underbrace{\sqrt{\lambda_k}}
_{(d_k-d_{k-1})\,\mathrm{times}}\right).
\label{eq:C9}
\end{align}
\end{reptheorem}

\begin{proof}
Calculating the commutator between $S_2$ and $S_3$ and using the fact that $S_2^{(k,k+1)}=-i S_3^{(k,k+1)}$ (see 
Lemma~\ref{lem1}) one obtains
\begin{align}\nonumber
&[S_2,S_3]=2 i\left( \sum_k\ketbra{k}{k}\otimes S_3^{(k+1,k)}S_3^{(k,k+1)}-\right.\\ 
&\left.\ketbra{k+1}{k+1}\otimes S_3^{(k,k+1)}S_3^{(k+1,k)}\right)
\label{eq:C10}
\end{align}

As $[S_2,S_3]=i c H$ and assuming that $n$ is even, one obtains the following set of equations
\begin{align}\nonumber
&S_3^{(2,1)}S_3^{(1,2)}= \frac{c\lambda_1}{2}\one_{d_1}\\ \nonumber
&S_3^{(3,2)}S_3^{(2,3)}-S_3^{(1,2)}S_3^{(2,1)}=\frac{c\lambda_2}{2}\one_{d_2}\\ \nonumber
&\:\:\:\:\:\:\:\:\:\:\:\:\:\:\:\:\:\:\:\:\:\:\:\:\:\:\:\:     \vdots\\ \nonumber
&S_3^{(n/2+1,n/2)}S_3^{(n/2,n/2+1)}-\\ \nonumber
&S_3^{(n/2-1,n/2)}S_3^{(n/2,n/2-1)}=\frac{c\lambda_{n/2}}
{2}\one_{d_{n/2}}\\ \nonumber
&S_3^{(n/2+2,n/2+1)}S_3^{(n/2+1,n/2+2)}-\\ \nonumber
&S_3^{(n/2,n/2+1)}S_3^{(n/2+1,n/2)}=\frac{c\lambda_{n/2+1}}
{2}\one_{d_{n/2+1}}\\ \nonumber
&\:\:\:\:\:\:\:\:\:\:\:\:\:\:\:\:\:\:\:\:\:\:\:\:\:\:\:\:     \vdots\\ 
&-S_3^{(n-1,n)}S_3^{(n,n-1)}= \frac{c\lambda_n}{2}\one_{d_n}
\label{eq:C11}
\end{align}

Using the singular value decomposition of $S_3^{(k,l)}$, define the unitary matrices 
$W^{(k,l)}:\cH_{d_l}\to\cH_{d_l}$, and $V^{(k,l)}: \cH_{d_k}\to\cH_{d_k}$, such that 
\begin{equation}
S_3^{(k,l)}=W^{(k,l)}D_3^{(k,l)}V^{(k,l)^\dagger},
\label{eq:C12}
\end{equation}
where $D_3^{(k,l)}$ is a $d_l\times d_k$ matrix containing the singular values of $S_3^{(k,l)}$ along its diagonal and 
zeros everywhere else.  Then the equations in Eq.~\eqref{eq:C11} read 
\begin{align}\nonumber
&W^{(2,1)}D_3^{(2,1)}D_3^{(1,2)}W^{(2,1)\dagger}=\frac{c\lambda_1}{2}\one_{d_1}\\ \nonumber
&W^{(3,2)}D_3^{(3,2)}D_3^{(2,3)}W^{(3,2)\dagger}-\\ \nonumber
&V^{(2,1)}D_3^{(1,2)}D_3^{(2,1)}V^{(2,1)\dagger}=\frac{c\lambda_2}{2}\one_{d_2}\\ \nonumber
&\:\:\:\:\:\:\:\:\:\:\:\:\:\:\:\:\:\:\:\:\:\:\:\:\:\:\:\:     \vdots\\ 
&-V^{(n,n-1)}D_3^{(n-1,n)}D_3^{(n,n-1)}V^{(n,n-1)^\dagger}=\frac{c\lambda_{n}}{2}\one_{d_n}.
\label{eq:C13}
\end{align}

For $k<n$, multiplying the $k^{\mathrm{th}}$ equation in Eq.~\eqref{eq:C13} from the left by $W^{(k+1,k)\dagger}$ and 
on the right by $W^{(k+1,k)}$ gives  
\begin{align}\nonumber
&D_3^{(2,1)}D_3^{(1,2)}=\frac{c\lambda_1}{2}\one_{d_1}\\ \nonumber
&D_3^{(3,2)}D_3^{(2,3)}-U^{(2,1)}D_3^{(1,2)}D_3^{(2,1)}U^{(2,1)\dagger}=\frac{c\lambda_2}{2}\one_{d_2}\\ \nonumber
&\:\:\:\:\:\:\:\:\:\:\:\:\:\:\:\:\:\:\:\:\:\:\:\:\:\:\:\:\:\:\:\:\:\:\:\:\:\:\:\:\:\:   \vdots\\ 
&-V^{(n,n-1)}D_3^{(n-1,n}D_3^{(n,n-1)}V^{(n,n-1)^\dagger}=\frac{c\lambda_{n}}{2}\one_{d_n}.
\label{eq:C14}
\end{align}
where $U^{(k+1,k)}=W^{(k+1,k)^\dagger} V^{(k,k-1)}$.  One possible solution for Eq.~\eqref{eq:C14} is given by choosing  $U^{(k+1,k)}=\one_{d_{k}},\, \forall k(1,n-1)$ and $V^{(n,n-1)}=\one_{d_n}$.  

As $D_3^{(2,1)}$ is a $d_1\times d_2$ diagonal matrix, it follows that in order for the first equation in Eq.~\eqref{eq:C14} to hold it is necessary that $d_2\geq d_1$ as otherwise $D_3^{(2,1)}D_3^{(1,2)}$ would contain $d_1-d_2$ zeros in the diagonal.  Moreover
\begin{equation}
D_3^{(1,2)}D_3^{(2,1)}=\frac{c}{2}\,\mathrm{diag}(\underbrace{\lambda_1,\ldots\lambda_1}_{d_1\, \mathrm{times}},\underbrace{0,\ldots,0}_{(d_2-d_1)\,\mathrm{times}}),
\label{eq:C15}
\end{equation}
which, upon substituting in the second equation of Eq.~\eqref{eq:C14}, gives
\begin{equation}
D_3^{(3,2)}D_3^{(2,3)}=\frac{c}{2}\,\mathrm{diag}(\underbrace{\lambda_1+\lambda_2,\ldots,\lambda_1+\lambda_2}_{d_1\, \mathrm{times}},\underbrace{\lambda_2,\ldots,\lambda_2}_{(d_2-d_1)\, \mathrm{times}}).
\label{eq:C16}
\end{equation}
As $D_3^{(3,2)}$ is a $d_2\times d_3$ diagonal matrix, the above equation implies that $d_3\geq d_2$ as otherwise 
$D_3^{(3,2)}D_3^{(2,3)}$ will contain $d_2-d_3$ zeroes.  Proceeding 
recursively one finds that for $1\leq k\leq n/2$, $d_{k+1}\geq d_k$ and that $D_3^{(k+1,k)}D_3^{(k,k+1)}$ is a 
$d_k\times d_k$ diagonal matrix whose first $d_1$ elements are equal to $\frac{c}{2}\sum_{i=1}^{k}\lambda_k$, the 
next $d_2-d_1$ elements are equal to $\frac{c}{2}\sum_{i=2}^{k}\lambda_k$, the next $d_3-d_2$ elements are equal 
to $\frac{c}{2}\sum_{i=3}^{k}\lambda_k$ and so on until the last $d_k-d_{k-1}$ elements which are equal to $\frac{c}{2}\lambda_k$. 

Let us now consider the $(n/2+1)^{\mathrm{th}}$ equation in Eq.~\eqref{eq:C14} given by 
\begin{align}\nonumber
&D_3^{(n/2+2,n/2+1)}D_3^{(n/2+1,n/2+2)}-\\
&D_3^{(n/2,n/2+1)}D_3^{(n/2+1,n/2)}=\frac{c\lambda_{n/2+1}}{2}\one_{d_{n/2+1}}.
\label{eq:C17}
\end{align}
From the argument above $D_3^{(n/2+1,n/2)}D_3^{(n/2,n/2+1)}$ is a $d_{n/2}\times d_{n/2}$ diagonal matrix whose 
last $d_{n/2}-d_{n/2-1}$ are equal to $\frac{c}{2}\lambda_{n/2}$.  As $d_{n/2+1}\geq d_{n/2}$, it follows that  
$D_3^{(n/2,n/2+1)}D_3^{(n/2+1,n/2)}$ is a $d_{n/2+1}\times d_{n/2+1}$ diagonal matrix whose first $d_{n/2}$ elements are equal to the elements of $D_3^{(n/2+1,n/2)}D_3^{(n/2,n/2+1)}$ and the remaining $d_{n/2+1}-d_{n/2}$ elements are equal to zero. As $\lambda_{n/2+1}=-\lambda_{n/2}$ by assumption Eq.~\eqref{eq:C17} becomes
\begin{align}\nonumber
D_3^{(n/2+2,n/2+1)}D_3^{(n/2+1,n/2+2)}&=D_3^{(n/2,n/2+1)}D_3^{(n/2+1,n/2)}\\
&-\frac{c\lambda_{n/2}}{2}\one_{d_{n/2+1}}.
\label{eq:C18}
\end{align}
As $D_3^{(n/2+2,n/2+1)}D_3^{(n/2+1,n/2+2)}$ must be a positive semidefinite matrix it is necessary that  
$d_{n/2+1}=d_{n/2}$.  Proceeding recursively through the remaining equations in Eq.~\eqref{eq:C14} one establishes 
that in order for $D_3^{(n/2+k+1,n/2+k)}D_3^{(n/2+k,n/2+k+1)}$ to be a positive semidefinite matrix it is necessary that  
$d_k=d_{n+1-k}$ and that for $n/2+1\leq k\leq n$, $d_k\leq d_{k+1}$ thus proving the theorem.  A similar argument holds for the case where $n$ is odd.  This completes the proof.  
\end{proof}

\section{Eigenspectrum of nearest-neighbor Hamiltonian}
\label{append4}

In this appendix we determine the eigenvalues and corresponding multiplicities of the nearest-neighbor 
Hamiltonian in Eq.~\eqref{eq:nearest_neighbourH}.
\begin{observation}
Let 
\begin{equation}
H_{nn}=\sum_{i=1}^{n-1}\sz^{(i)}\sz^{(i+1)}.
\label{eq:D1}
\end{equation}
Then $\sigma(H_{nn})=\{\lambda_x=n-1-2x\,|\, x\in (0\ldots n-1)\}$ where each $\lambda_x$ has multiplicity 
given by 2$n-1\choose x$.
\end{observation}

\begin{proof}
As 
\begin{equation}
\sz=\ketbra{0}{0}-\ketbra{1}{1},
\label{eq:D2}
\end{equation}
Eq.~\eqref{eq:D1} reads
\begin{equation}
H_{nn}=\sum_{\bm m}\sum_{k=1}^{n-1}(-1)^{\left(\sum_{i=k}^{k+1}m_i\right)}\ketbra{\bm m}{\bm m},
\label{eq:D3} 
\end{equation}
where ${\bm m}\equiv m_1\ldots m_n$.  Hence, the spectrum of $H_{nn}$ is given by
\begin{align}\nonumber
\sigma(H_{nn})&=\left\{(-1)^{(m_1+m_2)}+(-1)^{(m_2+m_3)}+\ldots \right.\\
&\left.+(-1)^{(m_{n-1}+m_n)}|\, m_i\in(0,1),\,\forall i\in(1,\ldots,n)\right\}. 
\label{eq:D4}
\end{align}
Clearly the maximum eigenvalue of $H_{nn}$ is $n-1$, and occurs when $\forall i\in(1,\ldots,n),\, m_i=0\,
\mathrm{or}\, 1$. The lowest eigenvalue is $-(n-1)$ and occurs when $m_i+m_{i+1}=1,\,\forall i\in
(1,\ldots,n-1)$.  Furthermore, the second highest eigenvalue is $n-3$ and occurs when all but one of the 
summands in Eq.~\eqref{eq:D4} are positive and one is negative.  Similarly it follows that the spectrum of $H_{nn}$ is 
given by 
\begin{equation}
\sigma(H_{nn})=\{\lambda_x= n-1-2x,\, |\, x\in (0,\ldots, n-1)\}.
\label{eq:D5}
\end{equation}

To obtain the multiplicity of each eigenvalue one simply looks at the number of different combinations of 
summing positive and negative ones in order to yield a specific eigenvalue.  Each summand in Eq.~\eqref{eq:D4}
can be either $1$ or $-1$, with the former occurring when both $m_i,m_{i+1}$ are the same and the latter when
$m_i,m_{i+1}$ are different.  As $\lambda_0$ contains no negative summands it follows that the total number of ways 
of obtaining $\lambda_0$ is 2$n-1\choose 0$.  Similarly $\lambda_1$ can be obtained by choosing one out of the total 
of $n-1$ summands negative  and this can be done 2$n-1\choose 1$.  It is not hard to see that there are a total of 
2$n-1\choose x$ different ways to obtain $\lambda_x$.  This completes the proof. 
\end{proof}

\section{Upper bounds for quantum metrology under local dephasing channel}
\label{append5}

In this appendix we derive the upper bounds of 
Eqs.~(\ref{eq:local-dephasing-localH-bound},~\ref{eq:Escherboundnearestneighbour}) corresponding to phase 
estimation in the presence of local dephasing noise using a local and nearest-neighbor Hamiltonian respectively. To 
that end we determine $\Xi$ and $\Omega$ in Eq.~\eqref{eq:local-dephasing-escher-special-operators}, where 
$B=S_x$, and the operators $S_{\bm m}$ in Eq.~\eqref{eq:Krausn} are explicitly given by
\begin{equation}
 S_{\bm k}=p^{\frac{N-h(\bm k)}{2}}(1-p)^{\frac{h(\bm k)}{2}}\otimes_{i=1}^N\sz^{(k_i)},
 \label{E1}
\end{equation}
where $h(\bm k)$ is the Hamming weight of the binary vector $\bm k$ and $k_i$ is the $i^{\mathrm{th}}$ 
entry of $\bm k$.  Note that $S_{\bm k}$ is a diagonal matrix for all $\bm k$.  

As the matrix elements of $S_x$ in the computational basis are given by    
\begin{equation}
\left[S_x\right]_{\bm l,\bm k}=\begin{cases}
							    1, & \text{if }\bm k\in\{\sx^{(i)}\bm l\}_{i=1}^N\\
                                                            0, & \text{otherwise},
\label{E2}                                                     \end{cases}
\end{equation}
where the set $\{\sx^{(i)}\bm l\}_{i=1}^N$ contains all $N$-dimensional, binary vectors obtained from $\bm l$ by flipping 
one of its bits. Simple algebra gives 
\begin{equation}
\sum_{\bm l,\bm k}S_{\bm l}\left[S_x\right]_{\bm{lk}}S_{\bm k}=2p^{\frac{1}{2}}(1-p)^{\frac{1}{2}}S_z.
\label{E3}
\end{equation}
A similar calculation yields
\begin{equation}
\sum_{\bm l,\bm k}S_{\bm l}\left[S_x^2\right]_{\bm{lk}}S_{\bm k}=N\left(1-4p(1-p)\right)\one +4p(1-p)S_z^2
\label{E4}
\end{equation}
Plugging Eqs.~(\ref{E3},~\ref{E4}) into Eq.~\eqref{eq:local-dephasing-escher-lower-bound}, and recalling that $H=S_z$ yields Eq.~\eqref{eq:local-dephasing-localH-bound}.

For the case of the nearest-neighbor Hamiltonian (Eq.~\eqref{eq:nearest_neighbourH}) $\Xi$ and $\Omega$ read
\begin{align}\nonumber
\Xi&=2p^{\frac{1}{2}}(1-p)^{\frac{1}{2}}\left(\left\langle HS_z\right\rangle-\left\langle H\right\rangle\left\langle S_z\right\rangle\right)\\
\Omega&=N\left(1-4p(1-p)\right)\one +4p(1-p)\Delta S_z,
\end{align}
which, upon substituting into Eq.~\eqref{eq:local-dephasing-escher-lower-bound} yields 
Eq.~\eqref{eq:Escherboundnearestneighbour}.

\bibliographystyle{apsrev4-1}
\bibliography{parameterestimation}
\end{document}